%% file: apssamp.tex
\documentclass[
reprint,
superscriptaddress,
amsmath,
amssymb,
aps,
pra,
balancelastpage
]{revtex4-2}

\usepackage{graphicx}
\usepackage{dcolumn}
\usepackage{bm}
\usepackage{hyperref}
\usepackage{amssymb,amsmath,amsthm,amsfonts}
\usepackage{textgreek}
\usepackage{mathtools}
\usepackage{svg}
\usepackage[ruled,vlined,linesnumbered]{algorithm2e}
\usepackage{algpseudocode}
\usepackage{qcircuit}
\usepackage{physics}
\usepackage{xcolor}
\usepackage{mathrsfs}
\usepackage{bbm}
\usepackage{bm}
\usepackage{hhline}
\usepackage{cases}
\usepackage{etoolbox}
\usepackage{diagbox}
\usepackage{float}

\newcommand{\nocontentsline}[3]{}
\let\oldaddcontentsline\addcontentsline
\newcommand{\tocless}[2]{
  \let\addcontentsline\nocontentsline
  #1{#2}
  \let\addcontentsline\oldaddcontentsline}
  
\DeclareFontFamily{U}{matha}{\hyphenchar\font45}
\DeclareFontShape{U}{matha}{m}{n}{
<-6> matha5 <6-7> matha6 <7-8> matha7
<8-9> matha8 <9-10> matha9
<10-12> matha10 <12-> matha12
}{}
\DeclareSymbolFont{matha}{U}{matha}{m}{n}

\DeclareFontFamily{U}{mathx}{\hyphenchar\font45}
\DeclareFontShape{U}{mathx}{m}{n}{
<-6> mathx5 <6-7> mathx6 <7-8> mathx7
<8-9> mathx8 <9-10> mathx9
<10-12> mathx10 <12-> mathx12
}{}
\DeclareSymbolFont{mathx}{U}{mathx}{m}{n}

\DeclareMathDelimiter{\vvvert} {0}{matha}{"7E}{mathx}{"17}

\DeclarePairedDelimiterX{\normiii}[1]
{\vvvert}
{\vvvert}
{\ifblank{#1}{\:\cdot\:}{#1}}

\theoremstyle{plain}
\newtheorem{theorem}{Theorem}
\newtheorem{proposition}{Proposition}
\newtheorem{lemma}{Lemma}

\theoremstyle{definition}
\newtheorem{definition}{Definition}

\newtheorem{remark}{Remark}

\newtheorem{problem}{Problem}

\usepackage{thm-restate}

\newcommand{\eqn}[1]{\hyperref[eqn:#1]{(\ref*{eqn:#1})}}
\newcommand{\rem}[1]{\hyperref[rem:#1]{Remark~\ref*{rem:#1}}}
\newcommand{\thm}[1]{\hyperref[thm:#1]{Theorem~\ref*{thm:#1}}}
\newcommand{\cor}[1]{\hyperref[cor:#1]{Corollary~\ref*{cor:#1}}}
\newcommand{\defn}[1]{\hyperref[defn:#1]{Definition~\ref*{defn:#1}}}
\newcommand{\lem}[1]{\hyperref[lem:#1]{Lemma~\ref*{lem:#1}}}
\newcommand{\prop}[1]{\hyperref[prop:#1]{Proposition~\ref*{prop:#1}}}
\newcommand{\fig}[1]{\hyperref[fig:#1]{Fig.~\ref*{fig:#1}}}
\newcommand{\tab}[1]{\hyperref[tab:#1]{Table~\ref*{tab:#1}}}
\newcommand{\algo}[1]{\hyperref[algo:#1]{Algorithm~\ref*{algo:#1}}}
\renewcommand{\sec}[1]{\hyperref[sec:#1]{Section~\ref*{sec:#1}}}
\newcommand{\append}[1]{\hyperref[append:#1]{Appendix~\ref*{append:#1}}}
\newcommand{\fac}[1]{\hyperref[fac:#1]{Fact~\ref*{fac:#1}}}
\newcommand{\lin}[1]{\hyperref[lin:#1]{Line~\ref*{lin:#1}}}
\newcommand{\fnote}[1]{\hyperref[fnote:#1]{Footnote~\ref*{fnote:#1}}}

\newcommand{\prob}[1]{\hyperref[prob:#1]{Problem~\ref*{prob:#1}}}

\renewcommand{\d}{\mathrm{d}}

\DeclareMathOperator{\poly}{poly}
\DeclareMathOperator{\polylog}{polylog}

\DeclareMathOperator{\diag}{diag}

\SetKwInput{KwData}{Input}
\SetKwInput{KwResult}{Output}

\begin{document}

\preprint{APS/123-QED}

\title{Resource-efficient quantum simulation of transport phenomena via Hamiltonian embedding}
\author{Joseph Li}
\affiliation{Department of Computer Science, University of Maryland, College Park, Maryland 20742, USA}
\affiliation{Joint Center for Quantum Information and Computer Science, University of Maryland, College Park, Maryland 20742, USA}

\author{Gengzhi Yang}
\affiliation{Joint Center for Quantum Information and Computer Science, University of Maryland, College Park, Maryland 20742, USA}
\affiliation{Department of Mathematics, University of Maryland, College Park, Maryland 20742, USA}
\author{Jiaqi Leng}
\affiliation{Simons Institute for the Theory of Computing, UC Berkeley, California 94709, USA}
\affiliation{Department of Mathematics, UC Berkeley, Berkeley, California 94709, USA}
\author{Xiaodi Wu}
\email{xiaodiwu@umd.edu}
\affiliation{Department of Computer Science, University of Maryland, College Park, Maryland 20742, USA}
\affiliation{Joint Center for Quantum Information and Computer Science, University of Maryland, College Park, Maryland 20742, USA}

\date{\today}
\begin{abstract}
Transport phenomena play a key role in a variety of application domains, and efficient simulation of these dynamics remains an outstanding challenge.
While quantum computers offer potential for significant speedups, existing algorithms either lack rigorous theoretical guarantees or demand substantial quantum resources, preventing scalable and efficient validation on realistic quantum hardware.
To address this gap, we develop a comprehensive framework for simulating classes of transport equations, offering both rigorous theoretical guarantees---including exponential speedups in specific cases---and a systematic, hardware-efficient implementation.
Central to our approach is the Hamiltonian embedding technique, a white-box approach for end-to-end simulation of sparse Hamiltonians that avoids abstract query models and retains near-optimal asymptotic complexity.
Empirical resource estimates indicate that our approach can yield an order-of-magnitude (e.g., $42\times$) reduction in circuit depth given favorable problem structures.
We then apply our framework to solve linear and nonlinear transport PDEs, including the first experimental demonstration of a 2D advection equation on a trapped-ion quantum computer.
\end{abstract}

\maketitle
\input{main_body}
\newpage
\onecolumngrid

\tocless\bibliography{apssamp}
\appendix
\newpage

\tableofcontents
\input{appendix}

\end{document}

%% file: main_body.tex
\tocless\section{Introduction}
Transport phenomena are fundamental to understanding processes across all scales in nature, serving as a cornerstone in the study of fluid dynamics~\cite{anderson1995computational,tritton2012physical}, chemical reactions~\cite{rood1987numerical,fogler2020elements}, and thermodynamics~\cite{jepps2006thermodynamics}.
Partial differential equations (PDEs) are central to the modeling of transport dynamics. 
While numerical solutions to these models are often desired to facilitate scientific investigation or engineering design, they are largely limited to low-dimensional settings, as computational cost typically scales exponentially with model dimension.
This makes accurate modeling of high-dimensional transport phenomena, such as those in phase space~\cite{alves2022data}, parameter space~\cite{hesthaven2016certified}, or multi-species chemical diffusion~\cite{kee2017chemically}, prohibitively expensive for traditional approaches.

Quantum computing has provided an alternative yet potentially promising paradigm for simulating large-scale dynamics, including inherently quantum phenomena~\cite{feynman1982simulating} and complex classical processes~\cite{childs2017quantum,liu2021efficient}.
Given that transport dynamics typically involve a large number of degrees of freedom and exhibit wave-like behaviors, quantum simulation of such phenomena has become an active area of extensive study~\cite{brearley2024quantum, sato2024hamiltonian, sato2024quantumalgorithmpartialdifferential, zecchi2025improved, li2025potential}.

Existing quantum algorithms for simulating PDEs mainly fall into two categories. 
The first approach adopts variational quantum algorithms (VQAs) to prepare quantum states that represent solutions to differential equations~\cite{lubasch2020variational,kyriienko2021solving,demirdjian2022variational,over2024boundary, umer2024nonlinear,jaffali2024h}. While these methods are resource-efficient and have been demonstrated in real quantum hardware, they often fail to offer theoretical guarantees, leaving their scalability in practical problems unclear.
The second approach is based on quantum subroutines implemented on fault-tolerant quantum computers, such as Quantum Linear System Solvers (QLSS)~\cite{berry2014high, liu2021efficient, krovi2023improved, berry2024quantum} and Hamiltonian simulation~\cite{costa2019quantum, jin2023quantum, an2023linear, babbush2023exponential, brearley2024quantum, sato2024hamiltonian, wright2024noisy}. 
For certain structured dynamical simulations, these quantum algorithms can achieve provable (sometimes exponential) quantum speedups over classical algorithms. 
Nevertheless, realizing these advantages on near- and intermediate-term quantum hardware remains challenging, even for low-dimensional linear PDEs. A central bottleneck is the quantum input model: one must construct coherent circuits that encode the problem data (e.g., PDE coefficients). Existing approaches, such as sparse-input oracles~\cite{aharonov2003adiabatic}, block-encodings~\cite{gilyen2019quantum}, and QRAM~\cite{giovannetti2008quantum}, are often prohibitively costly and well beyond the capabilities of current quantum devices.

Recently, a new technique called Hamiltonian embedding~\cite{leng2024expanding} was proposed to mitigate the significant overhead incurred by sophisticated quantum input models, particularly in sparse Hamiltonian simulation.
In this case, the goal is to implement the unitary $e^{-itA}$ for a sparse Hermitian matrix $A$ over an evolution time $t$. Existing algorithms are often \textit{complexity-theoretically} efficient in the sense that they achieve high-accuracy simulation using only a polynomial number of black-box (oracle) queries to the matrix entries. However, the cost of implementing these oracles is often left implicit without careful analysis.
Hamiltonian embedding instead adopts a \emph{white-box} approach.
It explicitly maps the target sparse Hamiltonian $A$ into a larger qubit Hamiltonian (the \textit{embedding Hamiltonian}) composed solely of local Pauli operators. 
This resulting Hamiltonian can then be efficiently simulated on either digital or analog hardware. 
For structured families of sparse Hamiltonians, the embedding Hamiltonian can be realized using $n$ qubits and at most $n^2$ 2-qubit interactions, where $n$ is logarithmically small in the size of $A$, paving a pathway towards an exponential quantum advantage.

In this paper, we develop a framework for resource-efficient quantum simulation of transport dynamics by incorporating Hamiltonian embedding into quantum PDE solvers. 
Our framework offers three key advantages. (1) Rigorous performance guarantees: for certain classes of transport equations, we prove that our algorithms achieve exponential quantum speedups when implemented on fault-tolerant devices. (2) Hardware-aware compilation: the Hamiltonian embedding can be customized to the native operations supported by the target platform, enabling direct compilation down to the gate or even the pulse level. As a byproduct, we provide resource estimation in terms of gate count and circuit depth for various classes of PDEs and identify the most efficient implementation based on problem structures. (3) Real-machine validation: we implement our approach to simulate a two-dimensional advection equation on a trapped-ion quantum computer (IonQ Aria-1); to the best of our knowledge, this is the first real-hardware demonstration of its kind in the literature~\cite{xin2020quantum, demirdjian2022variational, sato2024hamiltonian, wright2024noisy}.

Our framework consists of two major steps, as illustrated in~\fig{schematic_new}.
First, since most transport PDEs are non-unitary processes, they need to be mapped to quantum evolution processes before they can be solved via quantum simulation.
In this work, we adopt the Schr\"odingerization technique~\cite{jin2023discrete}. 
Note that an additional linearization step is necessary for transport dynamics governed by nonlinear PDEs; see the discussion in Sec.~\ref{sec:nonlinear-model}.
Second, with spatial discretization, the Schr\"odingerized transport PDEs are reduced to sparse Hamiltonian simulation problems, which are implemented using quantum circuits via Hamiltonian embedding~\cite{leng2024expanding}. 

\begin{figure}[!ht]
    \centering
    \includegraphics[width=\linewidth]{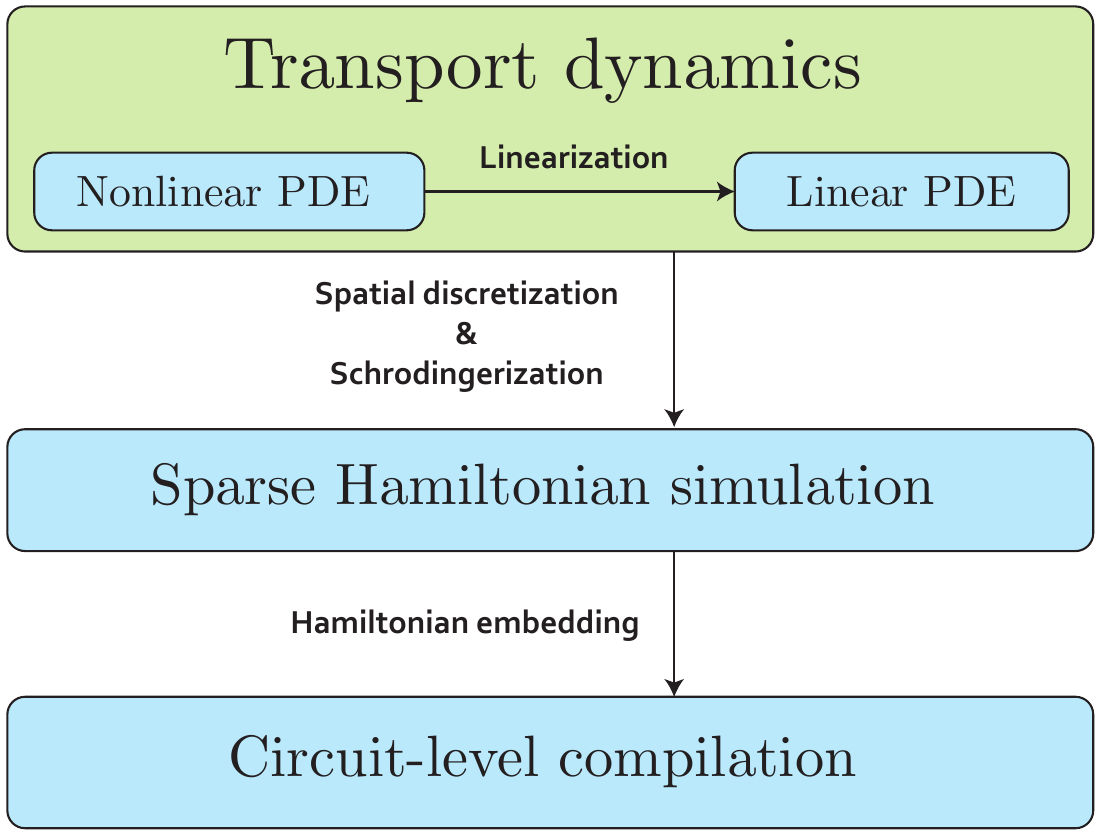}
    \caption{Schematic of quantum PDE solvers via Hamiltonian embedding.}
    \label{fig:schematic_new}
\end{figure}

Theoretically, we establish the exponential quantum speedup of our framework for PDEs whose spatial discretization admits a sparse, tensor product structure. This problem class encompasses classically intractable problems, such as the spatially varying advection equation and nonlinear scalar hyperbolic PDEs.
Our framework mainly relies on \algo{pde_solver}, a quantum ODE solver that integrates Hamiltonian embedding with Schr\"odingerization and Richardson extrapolation.
When applied to transport PDEs in $d$ spatial dimensions, our algorithm uses $\mathcal{O}(d)$ qubits and has gate complexity polynomial in $d$, indicative of an exponential speedup over classical mesh-based algorithms. More details of the algorithm, including the near-optimal parameter dependence in the evolution time and simulation error, are available in~\thm{complexity_analysis_informal}.

We showcase our framework to two fundamental classes of transport PDEs, including linear advection equations (Sec.~\ref{sec:linear-model}) and nonlinear scalar hyperbolic equations (Sec.~\ref{sec:nonlinear-model}).
These equations are ubiquitous in atmospheric and environmental sciences~\cite{rood1987numerical, chock1983comparison}, and form the foundation for more sophisticated models in fluid dynamics~\cite{leveque2002finite}.
In both cases, we discuss their implementation on quantum computers and conduct extensive resource analyses.
For linear advection equations, when incorporating gate-level parallelization, our Hamiltonian embedding approach (using the one-hot embedding) results in the shallowest circuits.
For nonlinear scalar hyperbolic PDE, we find that unary and one-hot embeddings outperform the standard binary encoding in both circuit depth and two-qubit gate counts by roughly an order of magnitude.

The paper is structured as follows. In Sec.~\ref{sec:background}, we provide the necessary background on Hamiltonian-simulation-based quantum PDE solvers and introduce the Hamiltonian embedding technique. We also briefly review existing results on real-machine demonstration of PDE simulations and their limitations.
In Sec.~\ref{sec:framework}, we develop the general framework of simulating transport PDEs via Hamiltonian embedding, with a complexity analysis on the asymptotic gate counts. 
In Sec.~\ref{sec:linear-model} and Sec.~\ref{sec:nonlinear-model}, we apply our quantum algorithms to linear advection equations and nonlinear hyperbolic PDEs, respectively. For each model of transport dynamics, we construct its corresponding embedding Hamiltonian, perform numerical simulations, and conduct a resource analysis based on the actual gate counts. We also demonstrate the simulation of a 2D advection equation using the IonQ Aria-1 processor.  
We conclude in Sec.~\ref{sec:conclusion} with a discussion.

\tocless\section{Background}\label{sec:background}

\tocless\subsection{Schr\"odingerization}
Schr\"odingerization is a technique that transforms linear, non-unitary dynamics to a unitary PDE described by a Schr\"odinger equation.
Consider a time-evolving linear PDE in $d$ spatial dimensions.
To apply Schrodingerization, one discretizes in space to obtain a linear ODE:
\begin{equation}
    \frac{\d u(t)}{\d t} = Au(t), \qquad u(0)=u_0, \label{eqn:linear_ode_time_indep}
\end{equation}
where $A\in\mathbb{C}^{N^d \times N^d}$, $u(t)\in\mathbb{C}^{N^d}$, and $N$ is the number of grid points along each dimension.
The structure of $A$ depends on the PDE of interest and the spatial discretization scheme; with finite differences, $A$ is generally sparse.
Throughout, we express $A$ in terms of its Cartesian decomposition $A=H_1+iH_2$, where $H_1=\frac{A+A^\dagger}{2}$ and $H_2=\frac{A-A^\dagger}{2i}$ are both Hermitian.
For stability, $H_1$ is assumed to be negative semidefinite.
Note that this assumption is merely a technical one and can always be satisfied by introducing a global decay rate.

For many interesting PDEs, the real part $H_1$ is non-zero, leading to a non-unitary dynamics.
Schr\"odingerization leverages the warped phase transformation~\cite{jin2023quantum}: $u(t) \mapsto v(t,p) = e^{-p}u(t)$ (where $p>0$ is an auxiliary variable), which maps the non-unitary dynamics to a Schr\"odinger-type equation:
\begin{equation}
\begin{split}
    \frac{\partial}{\partial t} v(t,p) &= -H_1 \frac{\partial}{\partial p} v(t,p) + i H_2 v(t,p)\\
    v(0,p) &= e^{-|p|} u_0.\label{eqn:warped_phase_time_indep}
\end{split}
\end{equation}
By using a Hamiltonian simulation algorithm such as~\cite{berry2015hamiltonian}, the gate complexity of Schr\"odingerization was shown to be $\widetilde{\mathcal{O}}(T/\epsilon)$, assuming the discretization step size for $p$ is chosen as $\Delta p \sim \epsilon$.

It is worth noting that another closely related framework, namely, the linear combination of Hamiltonian simulation (LCHS)~\cite{an2023linear}, may also be used to map non-unitary processes to Hamiltonian evolution.
For near-term implementation, the hybrid implementation of LCHS requires the execution of many randomly chosen circuits due to an importance sampling procedure.
With the same total number of circuit executions, we find that both methods achieve similar performance (see \append{lchs_comparison}.
In this work, we focus on Schr\"odingerization since it only requires a single quantum circuit.

\tocless\subsection{Hamiltonian embedding}
Hamiltonian embedding reduces the problem of simulating a sparse Hamiltonian $H$ to that of simulating a local Hamiltonian of the form $\widetilde{H}=\sum_{j} c_j H_j$, where $H_j$'s are local operators~\cite{leng2024expanding}.
While such mappings are known to exist and form the cornerstone of Hamiltonian complexity theory~\cite{cubitt2018universal}, they have rarely been leveraged for algorithmic purposes. Hamiltonian embedding addresses this by providing explicit (i.e., ``white-box'') constructions of the \textit{embedding Hamiltonian} $\widetilde{H}$ through the introduction of ancilla qubits.
Under an appropriate (hardware-aware) choice of embedding scheme, the locality of $\widetilde{H}$ may be much lower than that of $H$.
In this regard, Hamiltonian embedding may be viewed as a time-space tradeoff that reduces locality at the cost of using a larger Hilbert space.
In the following, we describe perturbative Hamiltonian embedding and explicit constructions for certain sparse matrices.

Let $H$ be an $n$-dimensional Hermitian operator, denoted as the problem Hamiltonian.
Assume there exists a $q$-qubit operator $Q$ and subspace $\mathcal{S}$ such that $Q|_{\mathcal{S}}=H$.
Consider another $q$-qubit operator $H^{\mathrm{pen}}$ whose ground-energy subspace is precisely $\mathcal{S}$.
$H^{\mathrm{pen}}$ is called the \textit{penalty Hamiltonian} and $\mathcal{S}$ is called the \textit{embedding subspace}.
Then, one can construct an embedding Hamiltonian
\begin{equation}
    \widetilde{H}=gH^{\mathrm{pen}} + Q.
\end{equation}
Relative to $\mathcal{S}$ and $\mathcal{S}^\perp$, the off-diagonal blocks are denoted by $R\coloneqq P_{\mathcal{S}^\perp} \widetilde{H} P_{\mathcal{S}}$ (where $P_{\mathcal{S}}$ and $P_{\mathcal{S}^\perp}$ are projectors).
When $\|R\|$ is sufficiently small relative to $g$, the dynamics of $H$ are approximated by that of $\widetilde{H}$, restricted to $\mathcal{S}$.
Specifically, Ref.~\cite[Theorems 1 and 3]{leng2024expanding} shows that the simulation error is bounded as
\begin{equation}
    \left\|\left(e^{-i\widetilde{H}t}\right)\Big|_{\mathcal{S}} - e^{-iHt}\right\| \leq (2\eta\|\widetilde{H}\|+\epsilon)t,
\end{equation}
where $t$ is the evolution time, $\eta\sim\|R\|/g$ and $\epsilon\sim\|R\|/g$.

For applications, we not only require an embedding of $H$, but also embeddings of an observable $O$ and input state $\ket{\psi}$, all using the same embedding subspace $\mathcal{S}$.
Thus, Hamiltonian embedding may be viewed as a mapping of the triple $(H,O,\ket{\psi})$ to $(\widetilde{H}, \widetilde{O}, |\widetilde{\psi}\rangle)$.
We note that for an $n$-qubit Hamiltonian $H$, taking $\mathcal{S}=\mathbb{C}^{2^n}$ corresponds to the trivial embedding $(H,O,\ket{\psi})\mapsto (H,O,\ket{\psi})$.
Since each $n$-bit bitstring can be viewed as a binary representation of an integer $j\in\{0,\dots,2^{n-1}\}$, this embedding is known as the \textit{standard binary encoding}.

Explicit constructions of perturbative Hamiltonian embeddings were developed for embedding several sparse matrix structures.
In particular, these embeddings are well-suited for banded and banded circulant matrices, but may also be applied to general sparse matrices.
These embedding schemes are related to alternative encoding systems such as the unary and one-hot encodings, in contrast to the standard binary encoding.

The naive application of Hamiltonian embedding to a general unstructured $n\times n$ matrix may require $\mathcal{O}(n)$ qubits and $\mathcal{O}(n^2)$ terms, nullifying the potential of exponential quantum speedup.
However, Ref.~\cite{leng2024expanding} shows that the Hamiltonian embedding technique preserves tensor product structure of the problem Hamiltonian.
In other words, when $H$ has appropriate tensor product structures, $\widetilde{H}$ can be constructed using resources that scale as $\polylog(n)$.
This feature of Hamiltonian embedding allows the possibility of exponential quantum speedup, depending on the underlying application.

\tocless\subsection{Relevant work}
While early quantum PDE solvers relied on quantum linear system solvers~\cite{liu2021efficient, childs2020quantum}, variational quantum algorithms have gained prominence due to their lower resource requirements.
Ref.~\cite{lubasch2020variational} uses multiple copies of variational quantum states to solve nonlinear PDEs, and a similar approach was applied to solve Burgers' equation~\cite{jaksch2023variational}.
Ref.~\cite{siegl2025tensor} develops a tensor-network based approach for simulating 1D linear Euler equations.
Thus far, these variational approaches have been limited to problems in low dimensions and lack evidence for provable quantum speedup.

Recent quantum algorithms for simulating PDEs with theoretical guarantees primarily leverage algorithms for Hamiltonian simulation.
Ref.~\cite{brearley2024quantum} presents an algorithm for solving the advection equation by embedding the non-unitary time-evolution operator within a Hamiltonian simulation at each time step.
Sato et al.~\cite{sato2024hamiltonian} showed an efficient decomposition of finite difference operators via the Bell basis, which were used to simulate simple advection and wave equations exhibiting unitary dynamics.
Subsequently, this approach was applied in conjunction with Schr\"odingerization to simulate non-unitary dynamics~\cite{hu2024quantum}.
While this strategy could only be applied to PDEs with constant coefficients, Ref.~\cite{sato2024quantumalgorithmpartialdifferential} developed an algorithm for simulating PDEs with spatially varying coefficients via logic minimization.
Despite these advances, the execution of these algorithms requires large fault-tolerant quantum computers and experimental demonstrations remain limited (a simulation of the 1D wave equation on 3 qubits was presented in Ref.~\cite{sato2024hamiltonian}).

\tocless\section{Quantum PDE Solvers via Hamiltonian Embedding}\label{sec:framework}
Our general framework for simulating transport dynamics is based on the following problem:
\begin{problem}[Simulation of $k$-th order linear PDEs]\label{prob:linear_pdes}
    Let $d\geq 1$, $x\in\mathbb{R}^{d}$, and consider the linear PDE
    \begin{equation}
    \begin{split}
        &\frac{\partial y}{\partial t} = F\left(x,y,D_x y,\dots,D_x^{k}y\right),\\
        &y(0,x) = y_0(x)
    \end{split}\label{eqn:linear_pde}
    \end{equation}
    where $y\colon \mathbb{R}^{+}\times \mathbb{R}^{d} \to \mathbb{R}$ and $F\colon \mathbb{R}^d\times \mathbb{R}\times \mathbb{R}^{d} \times \dots \times \mathbb{R}^{d^k} \to \mathbb{R}$ is linear.
    Given an observable $f(x)\colon \mathbb{R}^d\to \mathbb{R}$, evolution time $T>0$, and precision $\epsilon>0$, we consider the task of computing an $\epsilon$-close approximation to $\int_{x\in\mathbb{R}^{d}} \d x\,y(T,x)^\dagger f(x)y(T,x)$.
\end{problem}

\prob{linear_pdes} describes the simulation of general $k$-th order linear PDEs, which includes examples such as advection-diffusion equations, the Black-Scholes equation, or Maxwell's equation.
Through linear representations of nonlinear PDEs, it also extends to some non-linear PDEs such as scalar hyperbolic PDEs or Hamilton-Jacobi~\cite{jin2003level}.
It does not contain other nonlinear PDEs such as the viscous Burgers' equation or the nonlinear reaction-diffusion equation (which instead may be solved via Carleman linearization~\cite{liu2021efficient, liu2023efficient}).

In this paper, we focus on transport dynamics in $d$ spatial dimensions, which could be linear or non-linear depending on the application.
In case of nonlinearities, we first apply a linearization procedure to map the nonlinear PDE to a linear PDE in the form~\eqn{linear_pde}; the linearization procedure is standard~\cite{leyton2008quantum, joseph2020koopman, jin2022algorithms, liu2021efficient}, and we omit it here for brevity.

Given a linear PDE, we discretize in space to obtain a linear ODE system of dimension $N^d$, where $N$ is the number of grid points along each (spatial) dimension.
Then, the discretized PDE is mapped to a unitary process (described by a sparse Hamiltonian) via Schr\"odingerization~\cite{jin2023quantum}, followed by the Hamiltonian embedding~\cite{leng2024expanding} to obtain resource-efficient implementation on near- and intermediate-term devices.
The overall algorithm is presented in \algo{pde_solver}.

\begin{center}
\begin{algorithm}[ht!]
\caption{Quantum PDE solvers via Hamiltonian embedding}\label{algo:pde_solver}
\KwData{Coefficient matrix $A=H_1+iH_2$, observable $O$, evolution time $T$, precision $\epsilon$}
\KwResult{Estimate of $u(t)^\dagger O u(t)$, where $u(t)$ is the solution to \eqn{linear_ode_time_indep2}.}
For $k=1,\dots,d$, choose an embedding for the terms $h_{j,k}^{(1)}$ and $h_{j,k}^{(2)}$ comprising $H_1$ and $H_2$ in \eqn{sparse_matrices}.\\
Use Hamiltonian embedding to map $(H, O, u(0))$ to $(\widetilde{H}, \widetilde{O}, \widetilde{u}(0))$.\\
Apply Schr\"odingerization to simulate the embedded ODE:
\begin{align}
    \frac{\d \widetilde{u}(t)}{\d t} = \widetilde{H}\widetilde{u}(t),\qquad \widetilde{u}(0) = \widetilde{u}_0,\label{eqn:embedded_ode}
\end{align}
using a $q$-th order staged product formula for Hamiltonian simulation and Richardson extrapolation with error tolerance $\epsilon$.\\
Return the $\epsilon$-accurate estimate of $\widetilde{u}(T)^\dagger \widetilde{O}\widetilde{u}(T)$.
\end{algorithm}
\end{center}

\tocless\subsection{Spatial discretization and Schr\"odingerization}

In this work, we consider Eulerian discretizations using either finite differencing schemes or the Fourier spectral method.

When $N$ grid points are used along each dimension, the discretization error along each dimension is typically $\epsilon=\mathcal{O}(1/N)$ using finite differences (for instance, using the upwind scheme).
For a PDE with $d$ spatial variables, this results in $N=\mathcal{O}(d/\epsilon)$ to achieve an overall error $\epsilon$.
When periodic boundary conditions are imposed along a certain dimension, we may instead apply the Fourier spectral method.
The advantage of spectral methods over finite difference schemes is that they result in spectral convergence, i.e. the error decreases faster than $O(1/N^k)$ for all integers $k$~\cite{boyd2001chebyshev}.
We note that most suitable discretization scheme depends on the structure of the problem and that there is no single approach that works best in general.

Upon spatial discretization of \eqn{linear_pde} with $N$ grid points along each dimension, we obtain the $N^d$-dimensional linear ODE:
\begin{align}
    \frac{\d u(t)}{\d t} = Au(t),\qquad u(0) = u_0,\label{eqn:linear_ode_time_indep2}
\end{align}
where $A\in \mathbb{C}^{N^d \times N^d}$ and $u\in\mathbb{C}^{N^d}$.
If $O$ is a discretization of $f(x)$, then we can solve \prob{linear_pdes} by simulating \eqn{linear_ode_time_indep2} on quantum computers and estimate the quantity $u(T)^\dagger O u(T)$.

As we will see in \sec{ham_embed}, the use of Eulerian discretization schemes on regular grids allows us to exploit the tensor-product structure of $A$, which is the key to achieving exponential quantum speedup via Hamiltonian embedding.
While other discretization methods (such as finite volume or finite element) are commonly used in classical numerical simulations~\cite{leveque2002finite, brenner2008mathematical}, proposed quantum algorithms for finite volumes~\cite{chen2021quantum} or finite elements~\cite{clader2013preconditioned, montanaro2016quantum, deiml2025quantum} require sophisticated input models such as QRAM or block encodings due to the lack of a regular grid or mesh structure.

Having mapped the PDE to a finite-dimensional linear ODE, we apply Schr\"odingerization to embed the dynamics within a Hamiltonian simulation problem.
The overall quantum circuit for Schr\"odingerization is depicted in~\fig{schrodingerization_circuit}.

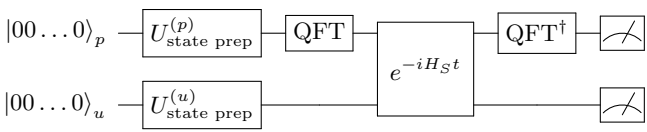
\begin{figure}[ht!]
    \centering
    \quad\quad\qquad
    \Qcircuit @C=1em @R=1em @!R {
        \lstick{\ket{00\dots0}_p} & \gate{U_{\text{state prep}}^{(p)}} & \gate{\text{QFT}} & \multigate{1}{e^{-iH_{S}t}} & \gate{\text{QFT}^{\dagger}} & \meter\\
        \lstick{\ket{00\dots0}_u} & \gate{U_{\text{state prep}}^{(u)}} & \qw & \ghost{e^{-iH_{S}t}} & \qw & \meter
    }
    \caption{Circuit diagram for Schr\"odingerization. The unitaries $U_{\text{state prep}}^{(p)}$ and $U_{\text{state prep}}^{(u)}$ denote state preparation circuits for the auxiliary variable $p$ and the solution $u$, respectively. 
    }
    \label{fig:schrodingerization_circuit}
\end{figure}
In essence, the main computational cost is in simulating the Hamiltonian
\begin{equation}
    H_S = H_1 \otimes H_{\mathcal{F}} - H_2 \otimes I,
\end{equation}
where $H_1 + iH_2 = A$ and $H_{\mathcal{F}}$ is a diagonal matrix.
Further details on the circuit implementation are provided in \append{schrodingerization_implementation}.

\tocless\subsection{Quantum simulation via Hamiltonian embedding}\label{sec:ham_embed}

Next, we employ the Hamiltonian embedding technique to simulate the Schr\"odingerized linear ODE.
\begin{definition}[Hamiltonian embeddings]\label{defn:ham_ebd}
Let $H \in \mathbb{C}^{N \times N}$ be a Hermitian operator, let $O\in\mathbb{C}^{N\times N}$ be a Hermitian observable, and $\ket{\psi} \in \mathbb{C}^{N}$ be a quantum state.
Let $\mathcal{S}$ be a subspace of $\mathbb{C}^{2^q}$ for some integer $q\geq 1$ and let $P_{\mathcal{S}}:\mathbb{C}^{N}\to\mathcal{S}$ be an isometry.
We say the triple $(\widetilde{H},\widetilde{O},|\widetilde{\psi}\rangle)$ is a Hamiltonian embedding of $(H,O,\ket{\psi})$ with embedding subspace $\mathcal{S}$ if $P_{\mathcal{S}}\widetilde{H} P_{\mathcal{S}}^\dagger = H$, $P_{\mathcal{S}}\widetilde{O}P_{\mathcal{S}}^\dagger = O$, $P_\mathcal{S}|\widetilde{\psi}\rangle=\ket{\psi}$, and $\mathcal{S}$ is an invariant subspace of $\widetilde{H}$.
\end{definition}
\begin{remark}
    \defn{ham_ebd} is a special case of the more generally defined (i.e., perturbative) Hamiltonian embeddings in~\cite{leng2024expanding}.
    Our definition allows us to avoid introducing additional simulation error, leading to a more simplified presentation.
\end{remark}

The Hamiltonian embedding technique provides a hardware-aware mapping of the \textit{problem Hamiltonian} $H$ to an \textit{embedding Hamiltonian} $\widetilde{H}$.
When the mapping is chosen such that $\widetilde{H}$ is compatible with hardware-native operations, $\widetilde{H}$ admits a decomposition which leads to a resource-efficient implementation on hardware.
Such a decomposition is essential for implementation on digital quantum computers.
Alternatively, it may be the case that $\widetilde{H}$ can be directly simulated by the native device Hamiltonian of an analog quantum simulator; however, analog quantum simulators typically require geometrically local connectivity which poses a major difficulty for the simulation of high-dimensional PDEs.
We leave a more detailed study of analog simulation of PDEs to future work (see \cite{jin2024analog} for a study using continuous variable quantum systems).

More concretely, we use and extend the one-hot, unary, and circulant unary embeddings developed in~\cite{leng2024expanding} to embed sparse matrices that commonly arise in the simulation of differential equations (see \append{ham_ebd} for details).
On the other hand, directly decomposing $H$ (which may be viewed as a trivial embedding) is equivalent to using the standard binary encoding.
Further discussions of the use of the standard binary encoding for representing certain diagonal and tridiagonal matrices are provided in \append{std_binary_enc}.

When discretizing a linear PDE via finite differences, the differential operators are mapped to sparse matrices which act nontrivially only along a single dimension.
For instance, a first-order differential operator acting on a variable $x_j$ is mapped as follows:
\begin{equation}
    \frac{\partial y}{\partial x_j} \mapsto I \otimes \dots \otimes \underbrace{A_j}_{\text{$j$-th term}} \otimes \dots \otimes I,
\end{equation}
where $A_j$ is a spatially discretized differential operator acting on the register used to represent $x_j$.
This locality structure is essential for harnessing the exponential speedup offered by quantum computation.
Leveraging this fact, we assume that $A$ has Hermitian and anti-Hermitian parts $H_1$ and $iH_2$ that are succinctly represented by the following tensor product structure:
\begin{equation}
    H_1 = \sum_{j=1}^{\Gamma} \gamma_j^{(1)}\bigotimes_{k=1}^{d} h_{j,k}^{(1)}
    \quad \text{and} \quad
    H_2 = \sum_{j=1}^{\Gamma} \gamma_j^{(2)}\bigotimes_{k=1}^{d} h_{j,k}^{(2)},
    \label{eqn:sparse_matrices}
\end{equation}
where $\gamma_{j}^{(1)},\gamma_{j}^{(2)}\in\mathbb{R}$ and $h_{j,k}^{(1)}$ and $h_{j,k}^{(2)}$ are sparse $N\times N$ matrices (e.g., diagonal, banded matrices, etc.).

Given the tensor product structure of $H_1$ and $H_2$, we may apply Hamiltonian embedding to simulate linear ODEs.
For $j=1,2,\dots,d$, we choose an embedding scheme for the $j$-th dimension.
This enables us to embed each $h_{j,k}^{(1)}$ and $h_{j,k}^{(2)}$ as embedding Hamiltonians $\widetilde{h}_{j,k}^{(1)}$ and $\widetilde{h}_{j,k}^{(2)}$, where $\widetilde{h}_{j,k}^{(1)}$ and $\widetilde{h}_{j,k}^{(2)}$ can be further decomposed into a sum of multi-qubit operators.
As a result, the non-Hermitian matrix $H$ is mapped to $\widetilde{H}=\widetilde{H}_1+i \widetilde{H}_2$ where $\widetilde{H}_1$ and $\widetilde{H}_2$ are Hamiltonian embeddings of $H_1$ and $H_2$, respectively.

To illustrate Hamiltonian embedding on a concrete example, consider a two-dimensional setting where for $m\in\{1,2\}$, $H_m = I \otimes h_{1,2}^{(m)} + h_{2,1}^{(m)} \otimes I$ and all $h_{j,k}^{(m)}$ and $h_{j,k}^{(m)}$ are tridiagonal matrices.
Then by embedding $h_{j,k}^{(m)}$ as $\widetilde{h}_{j,k}^{(m)}$, we obtain the embedding Hamiltonians $\widetilde{H}_m = I \otimes \widetilde{h}_{1,2}^{(m)} +  \widetilde{h}_{2,1}^{(m)} \otimes I$.
More generally, the embedding Hamiltonians admit decompositions
\begin{equation}
    \widetilde{H}_1 = \sum_{\ell=1}^{L} \alpha_{\ell}^{(1)} P_{\ell}
    \qquad\text{and}\qquad
    \widetilde{H}_2 = \sum_{\ell=1}^{L} \alpha_{\ell}^{(2)} P_{\ell},\label{eqn:ebd_decompositions}
\end{equation}
where $\alpha_{\ell}^{(1)}, \alpha_{\ell}^{(2)} \in \mathbb{R}$ and $P_{\ell}$'s are multi-qubit operators that may be efficiently exponentiated on quantum computers~\footnote{While the Hamiltonian embeddings proposed in~\cite{leng2024expanding} yield $P_{\ell}^{(1)}$ and $P_{\ell}^{(2)}$ as multi-qubit Pauli operators, we note that other decompositions are possible. For instance, the Bell basis decomposition of finite difference operators \cite{sato2024hamiltonian} may be viewed as a special case of Hamiltonian embedding.
We assume that if $P_{\ell}$ is a $k$-qubit operator, then it may be exponentiated with $\mathcal{O}(k)$ 1- and 2-qubit gates, which is the case for both multi-qubit Pauli operators and Bell basis operators.}.

Given the decompositions~\eqn{ebd_decompositions}, we construct the embedding Hamiltonian
\begin{equation}
    \widetilde{H}_S = \widetilde{H}_1 \otimes H_\mathcal{F} - \widetilde{H}_2 \otimes I
\end{equation}
and simulate the dynamics via product formulas.

\tocless\subsection{Measurement and post-processing}

In addition to embedding $H_1$ and $H_2$, we embed the observable $O$ as $\widetilde{O}$ and assume access to an embedded initial state $\widetilde{u}(0)$, which when restricted to the embedding subspace, is proportional to $u(0)$.
After applying Schr\"odingerization, we measure the $p$-register and post-select the outcomes in which $p>0$.
Then up to a constant factor, $u(t)^\dagger O u(t)$ is estimated by computing the expectation of $\widetilde{O}$ with respect to the system register.

While simulation with a $p$-th order product formula allows for gate complexity scaling nearly linearly in evolution time, the scaling in terms of precision is only $\mathcal{O}(\epsilon^{-1/p})$.
To match the optimal $\mathcal{O}(\log(1/\epsilon))$ dependence achieved by quantum signal processing~\cite{low2017optimal}, we incorporate classical extrapolation and post-processing.

Richardson extrapolation has emerged as an approach to improve the gate complexity of product formulas beyond $\mathcal{O}(1/\epsilon^{p})$ for a $p$-th order product formula, while retaining desirable properties such as locality and commutator scaling~\cite{watson2025exponentially}.
Rather than estimating observable expectations using a fixed Trotter step size, the idea is to compute expectation values using a decreasing sequence of step sizes and extrapolate to the zero step size limit.
Under this approach, it has been shown that in terms of evolution time $T$ and precision $\epsilon$, the maximum Trotter number scales as $\mathcal{O}(T^{(1+1/p)}\log(1/\epsilon))$.
We leverage this approach to simulate the embedding Hamiltonian $\widetilde{H}_S$ with near-optimal gate complexity.
Further details of product formulas and Richardson extrapolation for Hamiltonian simulation are provided in \append{richardson_extrapolation}.

\tocless\subsection{Complexity Analysis}
We analyze the runtime of \algo{pde_solver}.
The main computational cost comes from the use of Hamiltonian embedding as a subroutine.
Using product formulas in conjunction with Richardson extrapolation, we achieve a gate complexity which scales nearly linearly in $T$ and polylogarithmically in $1/\epsilon$.

\begin{theorem}[PDE solvers using Hamiltonian embedding, informal]\label{thm:complexity_analysis_informal}
    Let $H_1$ and $H_2$ be Hamiltonians with Hamiltonian embeddings given by $\widetilde{H}_1=\sum_{\ell=1}^{L} \alpha_{\ell}^{(1)} P_\ell$ and $\widetilde{H}_2=\sum_{\ell=1}^{L} \alpha_{\ell}^{(2)} P_\ell$.
    Suppose $\widetilde{H}_1$ and $\widetilde{H}_2$ are both $k$-local, i.e. all $P_{\ell}$ act non-trivially on at most $k$ qubits.
    Then to estimate $u(T)^\dagger O u(T)$ with relative error $\epsilon$ using a $p$-th order $\Upsilon$-stage product formula, the total number of 1- and 2-qubit gates per circuit is
    \begin{equation}
        \mathcal{O}\left(kL\Upsilon^{(2+1/p)}(\lambda T)^{(1+1/p)} \log(1/\epsilon)\right),
    \end{equation}
    where $\lambda$ is a nested commutator whose full expression is given in \eqn{lamb_comm}.
\end{theorem}

We prove a more formal version of \thm{complexity_analysis_informal} in \append{complexity_analysis}.
Our result achieves quasi-linear-in-$T$scaling, which is characteristic of high-order product formulas.
While $p$-th order product formulas typically incur a gate complexity that scales as $\epsilon^{-1/p}$, the use of Richardson extrapolation allows us to improve the explicit dependence on $\epsilon$ to polylogarithmic in $1/\epsilon$.
Additionally, we retain commutator scaling in the prefactor $\lambda$ which we further bound in \append{complexity_analysis}.
We note that in \thm{complexity_analysis_informal}, the $\log(1/\epsilon)$ dependence only takes into account the error introduced by the quantum simulation algorithm.
In the context of PDE simulation, discretization errors may require the norms of $\widetilde{H}_1$ and $\widetilde{H}_2$ to scale with $1/\epsilon$ (i.e., using finite difference schemes).
As a result, the overall gate complexity may depend polynomially in $1/\epsilon$.
However, the end-to-end time complexity of any quantum algorithm for estimating observables is necessarily $\Omega(1/\epsilon)$, as dictated by the Heisenberg limit~\cite{zwierz2010general}.

\begin{figure*}
    \centering
    \includegraphics[width=1\textwidth]{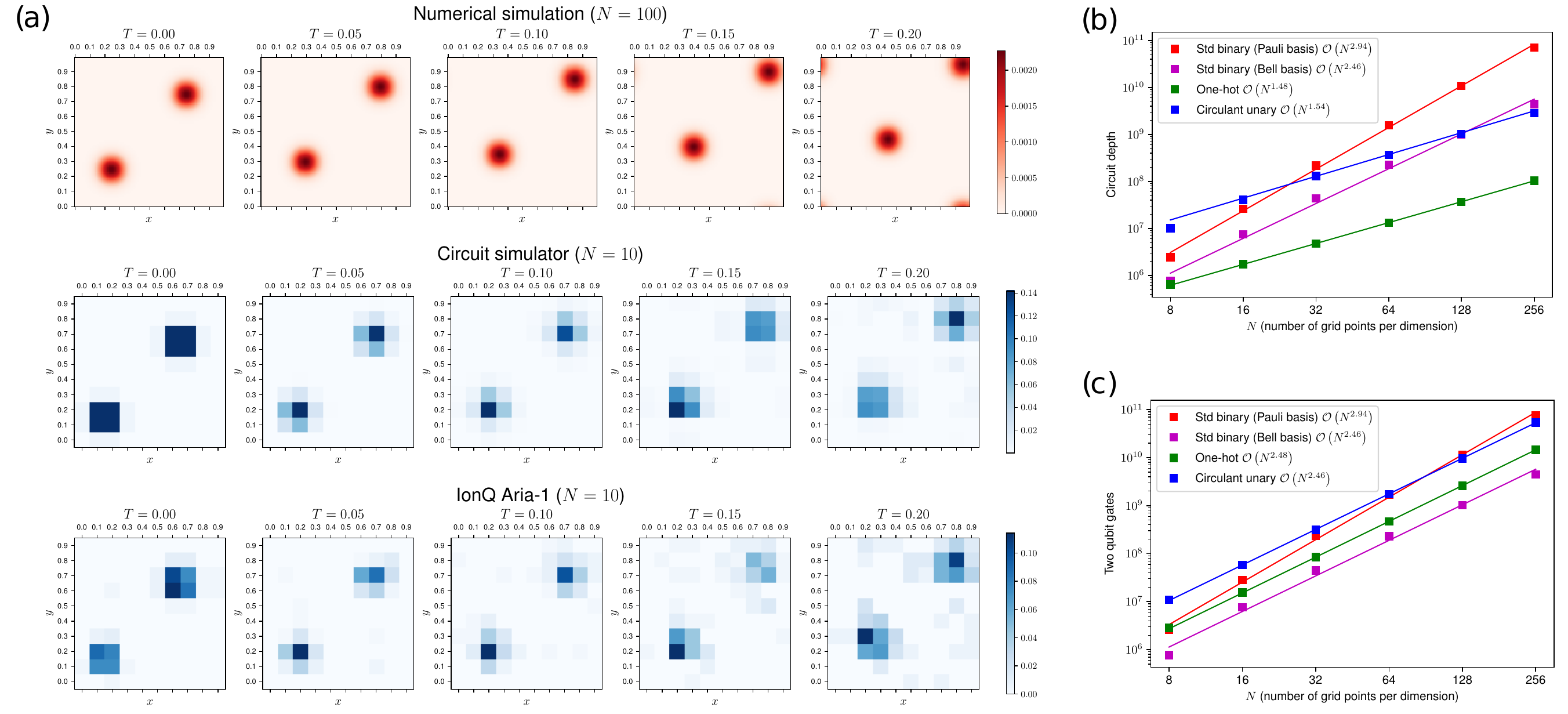}
    \caption{\textbf{Quantum simulation of the 2D linear advection equation.} 
    \textbf{Left:} (a) Each row shows the measured probability distribution for simulation times $T \in \{0, 0.05, 0.10, 0.15, 0.20\}$ (left to right).
    A high resolution classical simulation is shown in the top row for reference. The middle row shows results obtained from a noiseless circuit simulation, and the bottom row shows the results obtained from the Aria-1 quantum computer.
    \textbf{Right:} Empirical estimates of the two-qubit circuit depth (b) and gate count (c) required to simulate the advection equation with the upwind differences and Schr\"odingerization.}
    \label{fig:advection_equation_fig}
\end{figure*}

The choice of encoding scheme determines the structure (i.e., locality) of the embedding Hamiltonian and thereby affects the actual costs (i.e., gate counts and circuit depths) when implemented as quantum circuits.
In particular, the one-hot and unary encodings allow $N\times N$ finite difference operators to be embedded using $N$-qubit Hamiltonians with nearest-neighbor interaction terms.
In contrast, the standard binary encoding (for both the Pauli basis and Bell basis) decomposes finite difference operators using non-local operators on $\log_2(N)$ qubits (see \append{std_binary_finite_diff} for more details).
In this sense, Hamiltonian embedding offers a time-space tradeoff in which spatial resources are used to map the problem Hamiltonian to an embedding Hamiltonian whose dynamics can be more easily simulated on hardware.

The particular pattern of Hamiltonian embedding often enables a \textit{gate-level parallelization} that reduces the circuit depth through the use of additional ancilla qubits.
Note that Schr\"odingerization involves controlled Hamiltonian simulation, a naive implementation requires each controlled gate to be applied consecutively.
This overhead can be circumvented by encoding each ancilla qubit using a repetition code whose length is equal to the number of system (non-ancilla) qubits.
This way, the circuit depth can be reduced by a factor proportional to the number of qubits used.
In particular, Hamiltonian embeddings such as one-hot or unary allow for a circuit depth reduction by an $O(N)$ factor, while the standard binary encoding can only benefit by an $O(\log(N))$ factor.

\tocless\section{Linear transport dynamics}\label{sec:linear-model}
\tocless\subsection{Problem formulation}
Linear advection equations serve as prototypical examples of hyperbolic PDEs that model the transport of heat or fluids. 
Here, we consider a constant velocity advection equation in $d$ spatial dimensions given by
\begin{align}
    \frac{\partial u}{\partial t} + c\cdot \nabla u = 0\label{eqn:2d_adv_eq},\\
    u(0,x) = u_0(x),
\end{align}
where $c=(c_1,\dots,c_d)\in\mathbb{R}^{d}$ and $x\in [0,1]^d$ with periodic boundary conditions.

\tocless\subsection{Construction of Hamiltonian embedding}
To simulate the advection dynamics, we discretize in space using centered differences (see \append{discretization}) and obtain the linear ODE \eqn{linear_ode_time_indep2} with coefficient matrix
\begin{equation}
    A = -i\sum_{j=1}^{d} c_j(I\otimes \dots \otimes \underbrace{A_j}_{\text{$j$th term}} \otimes \dots \otimes I),
\end{equation}
where $A_j$'s are circulant tridiagonal matrices with $-i$ on the superdiagonal and $+i$ on the subdiagonal.
We use the one-hot embedding to embed each $A_j$, obtaining the embedding Hamiltonian
\begin{equation}\label{eqn:H-adv}
    H_{\text{adv}} = -\frac{1}{2h}\sum_{j=1}^{d}\sum_{k=1}^{N} c_j \left(\frac{X_{k+1}^{(j)} Y_{k}^{(j)} - Y_{k+1}^{(j)} X_{k}^{(j)}}{2}\right),
\end{equation}
where $N=10$, $h=1/N$, $d=2$, and $c_j=1$ for all $j=1,\dots,d$.
Since the use of centered differences results in unitary dynamics, Schr\"odingerization both straightforwardly reduce to Hamiltonian simulation without requiring additional ancilla qubits.

\tocless\subsection{Demonstration using quantum computers}
We validate our approach by a real-machine demonstration of an advection equation in $d=2$ spatial dimensions on the IonQ Aria-1 quantum computer~\footnote{The Aria-1 quantum computer is a trapped-ion device with all-to-all connectivity. At the time the circuits were executed, the SPAM error was 0.62\%, the 1-qubit gate fidelity was 0.9998, and the 2-qubit gate fidelity was 0.9871 as reported by IonQ's device characterization.}.
We choose the initial state to be two Gaussian peaks, and set the velocity vector to $(c_1,c_2)=(1,1)$.
We simulate $H_{\text{adv}}$ for time up to $T=0.2$ using the second-order Trotter-Suzuki formula. 
Due to the limitations of hardware noise, we do not perform Richardson extrapolation; instead, we apply the second-order Trotter-Suzuki formula and limit the number of Trotter steps to a maximum of two steps, resulting in a maximum gate count of 212 single-qubit gates and 115 two-qubit gates.
The real-machine demonstration results are shown in \fig{advection_equation_fig}.
Notably, despite a significantly lower resolution than the classical simulation, the experimental results clearly capture the direction and drift velocity of the Gaussian waves.

\tocless\subsection{Resource analysis}
We conduct a resource analysis for simulating \eqn{2d_adv_eq} using the upwind scheme.
This setting represents the typical scenario in which the dynamics are non-unitary and is simulated using Schr\"odingerization.

We compare the standard binary encoding (with the Pauli basis or Bell basis decompositions), the penalty-free one-hot embedding, and the penalty-free circulant unary embedding.
For all encoding schemes, we simulate the dynamics using the second-order Trotter-Suzuki formula.
We numerically compute the Trotter error for a single step for time $T/r$, then use the triangle inequality to estimate the error for $r$ Trotter steps.
Binary search is used to estimate the Trotter number required to reach an error below $\epsilon=5\times 10^{-2}$.
We use Qiskit~\cite{Qiskit} to optimize and compile the circuits to single qubit Pauli rotations and $XX$ rotation gates, using all-to-all device connectivity.
Furthermore, we employ the parallelization technique to optimize the circuit depth for the standard binary (Pauli basis), one-hot, and unary encodings.

Our estimates of the two-qubit circuit depth and gate counts are shown in \fig{advection_equation_fig} (b) and (c), respectively (see \append{explicit_counts} for explicit counts).
In terms of the circuit depth, the sparse encodings (one-hot and circulant unary) yield the most favorable asymptotic scaling (roughly $O(N^{3/2})$ from second-order Trotter).
Overall, the one-hot encoding yields the shallowest circuits, followed by the Bell basis for $N\leq 64$.
In particular, the one-hot encoding achieves a $42\times$ reduction in circuit depth for $N=256$.
The Bell basis tends to use the fewest two-qubit gates; however, all methods are within roughly an order of magnitude of each other.

\tocless\section{Nonlinear transport dynamics}\label{sec:nonlinear-model}

\begin{figure*}
    \centering
    \includegraphics[width=\textwidth]{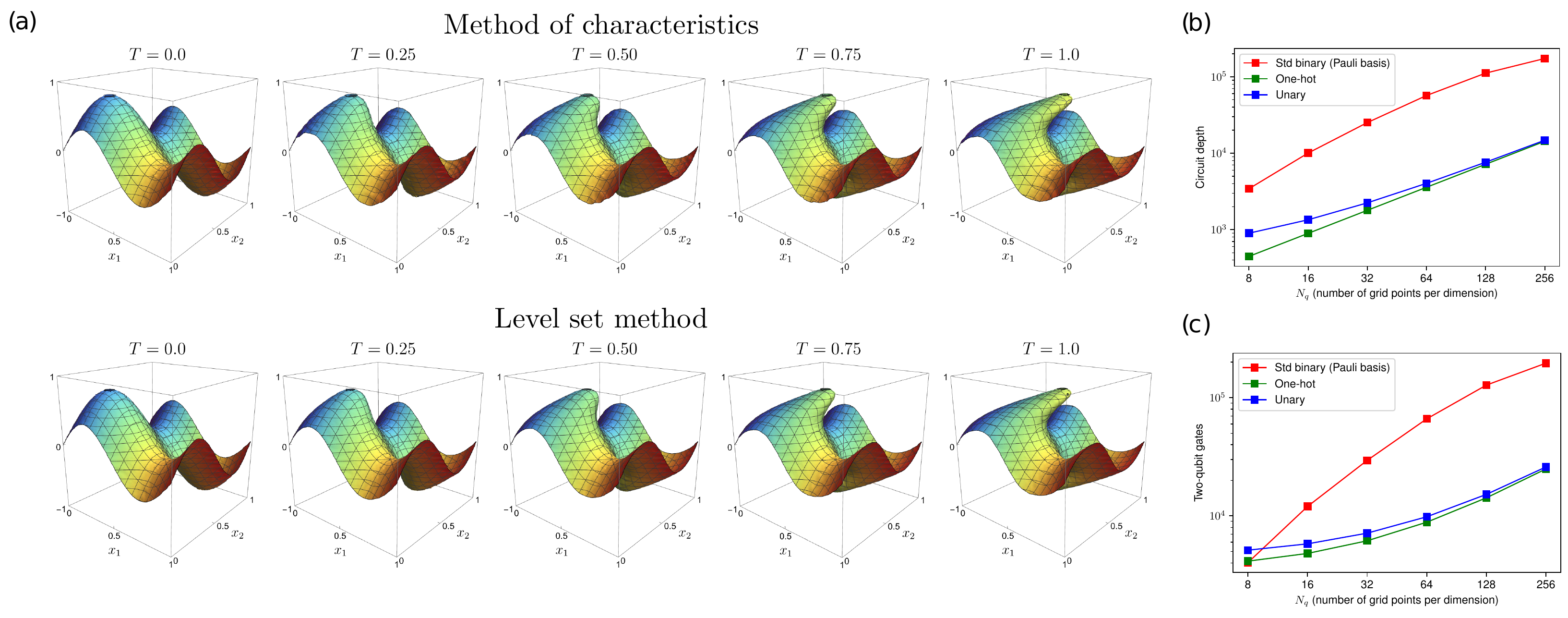}
    \caption{
    \textbf{Quantum simulation of nonlinear transport dynamics.}
    \textbf{Left: } (a) Numerical simulations of \eqn{nonlinear_pde} for times $T\in\{0,0.25,0.50,0.75,1.0\}$ using the method of characteristics (top) and the level set method (bottom). The initial condition is $u_0(x_1,x_2)=0.4(\sin(2\pi x_1)+\sin(2\pi x_2))$. We discretize using $N_x=N_q=128$ grid points for each dimension.
    \textbf{Right: } Empirical estimates of the two-qubit circuit depth (b) and gate counts (c) for simulating \eqn{nonlinear_pde} using $N_q$ grid points for $q$, where $N_q\in\{8, 16, 32, 64, 128\}$.}
    \label{fig:nonlinear_sim}
\end{figure*}

\tocless\subsection{Problem formulation}
Quantum computers are best suited for simulating linear dynamics due to the linearity of quantum mechanics.
However, many physical phenomena involve nonlinearity which poses a challenge for classical and quantum simulation alike.
Fortunately, certain classes of nonlinear PDEs can be mapped to linear dynamics and solved efficiently via quantum simulation.

Here, we consider a $(d+1)$-dimensional nonlinear scalar hyperbolic PDE given by
\begin{equation}
\begin{split}
    &\partial_t u + G(u) \cdot\nabla u = 0,\label{eqn:nonlinear_pde}\\
    &u(0,x) = u_0(x),
\end{split}
\end{equation}
where $x\in \mathbb{R}^d$ and $u : \mathbb{R}^{+}\times \mathbb{R}^{d}\to \mathbb{R}$.
This PDE captures nonlinearity through the coefficient function $G(u)$ and is a generalization of the inviscid Burgers' equation commonly studied in fluid dynamics.
To handle the nonlinearity, \eqn{nonlinear_pde} can be represented as a linear PDE using the level set method~\cite{jin2003level}. 
Introducing an auxiliary variable $q$, we apply the level set method and solve for $\phi(t,x,q)$ satisfying
\begin{equation}\label{eqn:level_set_eq}
\begin{split}
    &\frac{\partial \phi}{\partial t} + G(q) \cdot \nabla\phi = 0,\\
    &\phi(0,x,q) = \delta(q - u_0(x)),
\end{split}
\end{equation}
for which the solution is known to satisfy $\phi(t,x,q)=\delta(q-u(t,x))$~\cite{jin2022algorithms}.

For concreteness, we consider the numerical simulation of \eqn{level_set_eq} in two dimensions where $x=(x_1,x_2)\in[0,1]^2$, $q\in[-1,1]$, and $G(u)=u^3(1-u^4)$.
While finite differences may be used to discretize in space, the initial condition being a delta function results in poor numerical accuracy.
To avoid this issue, we instead use the Fourier spectral method in which we apply the Fourier transform to $x$ and simulate \eqn{level_set_eq} in the frequency basis.

On a quantum computer, our simulation of \eqn{level_set_eq} involves a quantum register for each of the variables $x_1$, $x_2$, and $q$.
We first apply the quantum Fourier transform to the registers for $x_1$ and $x_2$.
In the frequency basis, we simulate the diagonal Hamiltonian $H_d$ given by
\begin{equation}
    H_d = D_{x_1}\otimes I \otimes D_{q} + I \otimes D_{x_2} \otimes D_{q},
\end{equation}
where $D_{x_1}=D_{x_2}$ is a diagonal matrix of Fourier frequencies and $D_{q}$ is a diagonal matrix encoding values of $G(q)$.
Finally, we recover the numerical solution by applying inverse QFT to each of $x_1$ and $x_2$.
We note that the QFT and simulation of $H_d$ are all done in the appropriate embedding chosen for each of the variables $x_1$, $x_2$, and $q$.

\tocless\subsection{Construction of Hamiltonian embedding}

To embed the Hamiltonian $H_d$, we use the standard binary encoding of $x_1$ and $x_2$, while allowing freedom to choose an embedding for $q$.
The use of the standard binary encoding for $x_1$ and $x_2$ is justified by the fact that $D_{x_1}=D_{x_2} = \diag(0,1,\dots,N/2-1,-N/2,-N/2+1,\dots,-1)$ is a diagonal matrix of Fourier frequencies.
As shown in \append{std_binary_enc}, this matrix has a 1-local Pauli decomposition given by
\begin{equation}
     D_{x_1} = D_{x_2} = -\frac{1}{2}\left(I-2^{n} Z_n + \sum_{j=1}^{n} 2^{j-1} Z_j\right).
\end{equation}
Compared to sparse embedding such as one-hot or unary embeddings, the standard binary encoding is particularly suitable for representing $x_1$ and $x_2$ since it only requires $O(\log N)$ qubits while retaining a low locality Pauli decomposition.

Since $G(u)$ is a degree-4 polynomial, the best choice of embedding $q$ is not immediately clear.
Thus, we consider various embeddings of $D_q$ and analyze the resulting circuit depths.
As an example, applying the one-hot embedding to embed $D_q$ results in the embedding Hamiltonian
\begin{equation}
    \widetilde{D}_q=\sum_{j=1}^{N} G(q_j)\hat{n}_j^{(1)},
\end{equation}
where $\hat{n}_j^{(1)}=\frac{1}{2}\left(I-Z_{j}\right)$ and $q_j$ is the $j$-th grid point on the discretized interval $[-1,1]$.

The resulting overall embedding Hamiltonian is given by
\begin{equation}
    \widetilde{H}_d = \widetilde{D}_{x_1} \otimes I \otimes \widetilde{D}_q + I \otimes \widetilde{D}_{x_2} \otimes \widetilde{D}_q,
\end{equation}
where $D_{x_{j}}=\widetilde{D}_{x_{j}}$ due to the use of the standard binary encoding of $x_1$ and $x_2$.

\tocless\subsection{Numerical simulation}
Our numerical simulation of the level set method via the Fourier spectral method is shown in the bottom row of \fig{nonlinear_sim}.
We verify correctness by comparing against the method of characteristics, which are shown in the top row.
We emphasize that while the results of \fig{nonlinear_sim} are classically simulated, they are mathematically equivalent to a circuit-level simulation on quantum computers.

\tocless\subsection{Resource analysis}
We perform a resource analysis to compare the circuit depths required to simulate \eqn{level_set_eq} when using different embedding schemes.
We compare the standard binary, one-hot, and unary embeddings for $q$, while the standard binary code is used to encode $x_1$ and $x_2$.
Additionally, we parallelize the gates by introducing an ancilla qubit for each qubit used to represent $q$.
This is enabled by the fact that the simulation of $H_d$ can be viewed as Hamiltonian simulation of $D_q$ controlled on each qubit used to represent $x_1$ and $x_2$.

The estimated circuit depths and two-qubit gate counts are presented in \fig{nonlinear_sim} (b) and (c), respectively.
In general, the one-hot and unary embeddings outperform the standard binary code.
We emphasize that this is a direct consequence of the choice of embedding $\widetilde{D}_q$ for the matrix $D_q$.
For the unary and one-hot embeddings, $\widetilde{D}$ consists of identity and single-site Pauli-$Z$ operators.
On the other hand, the naive Pauli decomposition of $D_q$ generally contains $4$-local terms since $G(u)$ is a degree-4 polynomial.
As a result, simulating $\widetilde{D}_q$ using the unary and one-hot embeddings is significantly cheaper than a direct simulation of $D_q$, in terms of both circuit depth and gate counts.
This underscores the importance of choosing the right embedding for each problem, which can drastically influence the overall simulation cost.

\tocless\section{Discussion and Conclusions}\label{sec:conclusion}
In this paper, we developed a framework for resource-efficient simulation of transport dynamics.
Our use of Hamiltonian embedding enables the end-to-end implementation of differential equation solvers without sophisticated quantum circuit constructions.
Through extensive resource analyses, we find that certain embedding schemes are better suited for solving PDEs with certain mathematical structure.
In particular, sparse encodings such as the one-hot or unary codes are best suited for embedding finite difference operators and benefit most from parallelization of controlled Hamiltonian simulation.
Our approach fully addresses the difficulty of loading sparse Hamiltonians that arise from quantum simulation of transport dynamics with wide applicability to various PDE structures (e.g., spatially varying coefficients).
Furthermore, our use of product formulas and Richardson extrapolation make our algorithm suitable for intermediate-term and early fault-tolerant devices while retaining near-optimal asymptotic gate complexity.

While linear dynamics are most amenable to be solved via quantum simulation, our framework also encapsulates broader families of PDEs through the process of linearization. 
In this work, we explored the use of the level set method, which are limited to nonlinear scalar hyperbolic PDEs and Hamilton-Jacobi equations~\cite{jin2003level}.
It remains an open question whether other families of nonlinear PDEs offer linear representations that can then be solved via Hamiltonian simulation.

For practically useful applications of quantum differential equation solvers, there remain several limitations that remain unaddressed in our work.
For example, practical applications often require an initial state that is non-trivial to prepare.
While there are several proposed algorithms for preparing certain quantum states~\cite{grover2002creating, rattew2022preparing, mcardle2022quantum}, these methods either require coherent arithmetic or are formulated using sophisticated input models, making them impractical for intermediate-term quantum computers.
Developing more efficient initial state preparation schemes that are feasible for near-term quantum hardware remains relatively unexplored.
Furthermore, our real-machine demonstrations are limited in scale due to the rapid decoherence of NISQ-era devices.
As the capabilities of quantum hardware advance, we anticipate greater applicability of quantum differential equation solvers to real-world problems.

\section*{Acknowledgment} 
We thank Dong An for helpful discussions.
This work was partially funded by the U.S. Department of Energy, Office of Science, Office of Advanced Scientific Computing Research, Accelerated Research in Quantum Computing under Award Number DESC0020273 and DESC0025341, the Air Force Office of Scientific Research under Grant No. FA95502110051, the U.S. National Science Foundation grant CCF-1816695 and CCF-1942837 (CAREER), and a Sloan research fellowship. 
Jiaqi Leng is partially supported by DOE QSA Grant No. FP00010905, a Simons Quantum Postdoctoral Fellowship, and a Simons Investigator award through Grant No. 825053. 
We are also grateful to the access to IonQ machines provided by the National Quantum Laboratory (QLab) at UMD.

%% file: appendix.tex
\section{Finite difference schemes for hyperbolic PDEs}\label{append:discretization}
In this section, we review spatial discretization schemes for hyperbolic partial differential equations.

For illustrative purposes, we consider the discretization for the 1-dimensional variable coefficient case, in which
\begin{equation}\label{eqn:single_term}
    \frac{\partial u}{\partial t} = - f(x)\frac{\partial u}{\partial x}.
\end{equation}
Here we have not yet specified the boundary conditions.
The stability of the discretization scheme depends on the boundary conditions, and the chosen scheme ultimately determines the Hamiltonian to be simulated.
For all discretization schemes, the method of lines yields an ODE system
\begin{equation}\label{eqn:method_of_lines_ode}
    \frac{\d \mathbf{u}(t)}{\d t} = A\mathbf{u}(t),
\end{equation}
and $\mathbf{u}(t)=[u_1(t), u_2(t), \dots, u_N(t)]^\top$ denotes the vectorized mesh function and $A$ is a coefficient matrix which depends on the scheme chosen.

\subsection{Centered differences}
The centered differences scheme is stable only for periodic boundary conditions.
We discretize \eqn{single_term} using $N$ grid points and let $h = 1/N$ be the distance between grid points.
Let $u_j(t)$ denote the value of $u(t,x)$ at the $j$-th grid point upon applying discretization.
Let $q_1,\dots,q_{N}$ be equally spaced grid points such that $q_1 < \dots < q_N$, identifying $q_{N+1}$ with $q_{1}$. 
For symmetry, let $q_{j+1/2} = (q_j+q_{j+1})/2$ be the midpoint for the $j$-th subinterval.
By centered differences, we obtain the discretization scheme
\begin{equation}\label{eqn:centered_differences}
    \frac{\d u_{j}}{\d t} =- \left(\frac{f(q_{j+1/2}) u_{j+1} - f(q_{j-1/2}) u_{j-1}}{2h}\right) + \frac{f'(q_j)}{2}u_j + \mathcal{O}(h),
\end{equation}
for all $j$.
If $f(x)$ has the same periodic boundary condition as $u(x)$ does, thus \eqn{centered_differences} can be written as the ODE system \eqn{method_of_lines_ode} with coefficient matrix
\begin{equation}\label{eqn:centered_differences_matrix}
    A_c
    = \begin{pmatrix}
    \frac{f'(q_1)}{2} & -\frac{f(q_{1+1/2})}{2h} & & & \frac{f(q_{N+1/2})}{2h} \\
    \frac{f(q_{1+1/2})}{2h} & \frac{f'(q_2)}{2} & \ddots & & \\
     & \ddots & \ddots & \ddots & \\
     & & \ddots & \ddots & -\frac{f(q_{N-1/2})}{2h}\\
    -\frac{f(q_{N+1/2})}{2h} & & & \frac{f(q_{N-1/2})}{2h} & \frac{f'(q_N)}{2}
    \end{pmatrix}.
\end{equation}
For stability, $f'(x)$ should be negative.
If $f'(x)$ is bounded and $N$ is large, the diagonal terms are negligible.
In this case, one may simulate only the anti-Hermitian part $\frac{A-A^\dagger}{2}$ and approximate the dynamics by unitary evolution.

\subsection{Upwind scheme}
For hyperbolic PDEs, the dynamics are typically asymmetric, and the use of centered differences is often unstable.
An alternative approach with better stability is the upwind scheme, where a one-sided approximation to the derivative is used.
This also enables the use of non-periodic boundary conditions, i.e. an inflow boundary.
Note that the outflow boundary is already determined and an additional constraint cannot be imposed~\cite{leveque2007finite}.
Here, we assume the inflow boundary condition $u_1=0$ with coefficient function $f(x)>0$.
The upwind scheme in this case corresponds to the discretization of \eqn{single_term} by
\begin{equation}\label{eqn:upwind_scheme}
    \frac{\d u_{j}}{\d t} = - f(q_{j})\left(\frac{u_{j} - u_{j-1}}{h}\right).
\end{equation}
This corresponds to \eqn{method_of_lines_ode} with coefficient matrix
\begin{equation}\label{eqn:upwind_scheme_matrix}
    A_u
    = \begin{pmatrix}
    -\frac{f(q_1)}{h} & & & & \\
    \frac{f(q_{2})}{h} & -\frac{f(q_2)}{h} & & & \\
     & \ddots & \ddots & & \\
     & & \ddots & \ddots & \\
     & & & \frac{f(q_{N})}{h} & -\frac{f(q_N)}{h}
    \end{pmatrix}.
\end{equation}
Since $f(x)>0$, the eigenvalues of $\frac{A_u + A_u^\dagger}{2}$ are non-positive and thus the system is stable.
In the case when $f(x)<0$, one should instead approximate $\frac{\partial u(t,x_j)}{\partial x}$ by $\frac{u_{j+1}-u_{j}}{h}$.

\section{Hamiltonian embeddings of sparse matrices}\label{append:ham_ebd}
In this section, we review Hamiltonian embedding, a recently proposed framework that enables the simulation of sparse Hamiltonians \cite{leng2024expanding}.
For simulating partial differential equations, we focus on encoding schemes that allow for efficient embeddings of banded and circulant banded matrices, which arise from finite difference discretizations of differential operators.

To simulate a sparse matrix $A$, the Hamiltonian embedding technique makes use of an embedding Hamiltonian $H_A$ whose projection onto a subspace $\mathcal{S}$ precisely matches $A$.
That is, with respect to the subspaces $\mathcal{S}$ and $\mathcal{S}^\perp$, $H_A$ can be written as a block matrix
\begin{equation}
    H_A = \begin{pmatrix}
        A & R^\dagger \\
        R & G
    \end{pmatrix},
\end{equation}
where $A=P_{\mathcal{S}} H_A P_{\mathcal{S}}$, $G=P_{\mathcal{S}^\perp} H_A P_{\mathcal{S}^\perp}$, $R=P_{\mathcal{S}^\perp} H_A P_{\mathcal{S}}$, and $P_{\mathcal{S}}$ and $P_{\mathcal{S}^\perp}$ are projectors onto $\mathcal{S}$ and $\mathcal{S}^\perp$, respectively.

In the case when $\|R\|>0$, it was shown in \cite{leng2024expanding} that a sufficiently large penalty Hamiltonian may be used to control the simulation error.
However, the simulation of $H_A$ using usual product formulas leads to large simulation error when the penalty coefficient is large. 
To overcome this limitation, here we consider penalty-free versions of the one-hot, unary, and circulant unary embeddings which lead to improved asymptotic complexity using standard product formulas.

In this section, we index bitstrings starting from 1 and denote bitstrings from right to left, e.g. the first bit of $001$ is 1.
We use $X_j$, $Y_j$, and $Z_j$ to denote Pauli matrices 
\begin{align*}
X=\begin{pmatrix}
    0 & 1 \\
    1 & 0
\end{pmatrix}, \quad
Y=\begin{pmatrix}
    0 & -i \\
    i & 0
\end{pmatrix}, \quad
Z=\begin{pmatrix}
    1 & 0 \\
    0 & -1
\end{pmatrix}
\end{align*}
applied to the $j$-th qubit, respectively.
Furthermore, we make extensive use of the number operator $\hat{n}^{(1)}_j = \frac{I-Z_j}{2}$ and its orthogonal complement $\hat{n}^{(0)}_j = \frac{I+Z_j}{2}$ on the $j$-th qubit.
For a matrix $A$, we denote its $(j,k)$ entry by $A_{j,k}$ if $j\neq k$ or $A_j$ if $j=k$.
\subsection{Penalty-free embedding schemes}
\subsection*{One-hot embedding}
For simulating partial differential equations on digital quantum computers, the one-hot embedding without penalty offers both simplicity and flexibility due to the use of only 2-local operators.
\begin{definition}[One-hot code]
    The one-hot code consists of bitstrings $\{h_1,\dots,h_n\}$ such that for $j=1,\dots,N$, $h_j$ is the bitstring of length $N$ whose $j$-th bit is 1 and all other bits are 0.
\end{definition}

\begin{proposition}[One-hot embedding without penalty]
    Let $A\in \mathbb{C}^{N \times N}$ be a Hermitian matrix.
    Let $\mathcal{S}$ be the subspace spanned by the set of one-hot states $\{\ket{h_j} \mid j=1,\dots,N \}$.
    Then the Hamiltonian
    \begin{equation}
        H_A = \frac{1}{2}\left(\sum_{j<k} \Re(A_{j,k}) (X_k \otimes X_j + Y_k \otimes Y_j) + \Im(A_{j,k})(X_k \otimes Y_j - Y_k \otimes X_j)\right) + \left(\sum_{j=1}^{N}A_{j}\hat{n}^{(1)}_{j}\right)
    \end{equation}
    satisfies $H_A \big|_{\mathcal{S}} = A$ and $\mathcal{S}$ is an invariant subspace of $H_A$.
\end{proposition}
\subsection*{Unary embedding}
In the case of non-periodic boundary conditions, the finite difference operators are banded matrices, where the $(j,k)$ entries are nonzero only if $|j-k|\leq b$ for an integer bandwidth $b\geq 0$.
In this case, the unary embedding provides an alternative to the one-hot code.

\begin{definition}[Unary code]
    The unary embedding consists of bitstrings $\{u_1,\dots,u_{N}\}$ of length $(N-1)$ such that for $j=1,\dots,N$, the first $(j-1)$ bits of $u_j$ are 1's and remaining $(N-j)$ bits are 0's.
\end{definition}

Penalty-free embeddings of band matrices with arbitrary bandwidth may be constructed using the unary code.
However, the penalty-free version require increased locality and in general consists of a number of Pauli terms which is exponential in the bandwidth.
Therefore, for simplicity, we consider the penalty-free unary embedding of a tridiagonal matrix as follows.
\begin{proposition}[Unary embedding without penalty]\label{prop:unary_ebd}
    Let $A\in \mathbb{C}^{N\times N}$ be a Hermitian, tridiagonal matrix.
    Let $\mathcal{S}$ be the subspace spanned by the set of unary states $\{\ket{u_j}\ \mid j=1,\dots,N\}$. Then the Hamiltonian
    \begin{multline}
        H_A = \hat{n}_2^{(0)} \left(\Re(A_{1,2}) X_1 - \Im(A_{1,2}) Y_1 \right) + \left(\sum_{j=2}^{N-1} \hat{n}^{(0)}_{j+1} \left(\Re(A_{j,j+1})X_{j}-\Im(A_{j,j+1})Y_{j}\right) \hat{n}^{(1)}_{j-1} \right) \\
        + (\Re(A_{N-1,N}) X_{N} - \Im(A_{N-1,N}) Y_{n})\hat{n}^{(1)}_{N-1}  + \left(A_{1,1} I + \sum_{j=2}^{N} (A_{j} - A_{j-1})\hat{n}_{j-1}^{(1)} \right)
    \end{multline}
    satisfies $H_A \big|_{\mathcal{S}} = A$ and $\mathcal{S}$ is an invariant subspace of $H_A$.
\end{proposition}

\begin{remark}
    The unary embedding is equivalent to the one-hot embedding in the following sense.
    For any unary state $\ket{u_j}$, we can introduce an additonal ancilla qubit initialized to $\ket{1}$ and perform a ladder of CNOT gates:
    \begin{equation}
        \left(\prod_{j=2}^{N} CX_{j,j-1}\right)\ket{u_j}\ket{1} = \ket{h_j}.
    \end{equation}
    Consequently, any $(N-1)$-qubit state $\sum_{j=1}^{N}\alpha_j \ket{u_j}$ in the unary embedding subspace can be mapped to the $N$-qubit one-hot embedded state $\sum_{j=1}^{N}\alpha_j \ket{h_j}$ in depth $\mathcal{O}(N)$, and vice-versa.
    With mid-circuit measurement and feedforward, the CNOT ladder can be implemented in $\mathcal{O}(1)$ depth~\cite{baumer2025measurement}.
\end{remark}

\subsection*{Circulant unary embedding}
For matrices with cyclic structure (arising from the use of periodic boundary condition), the circulant unary embedding may be used.

\begin{definition}[Circulant unary code]
    Let $N$ be an even integer. The circulant unary code consists of bitstrings $\{c_1,\dots,c_N\}$ of length $N/2$ such that
    \begin{itemize}
        \item 
        if $j=1,\dots,N/2$, the first $(j-1)$ bits of $c_j$ are 1's and the remaining $(N+1-j)$ are 0's;
        \item 
        if $j=N/2+1,\dots,N$, $c_j$ is the bitwise complement of $c_{j-N/2}$.
    \end{itemize}
\end{definition}

Similar to the unary embedding, we consider the penalty-free circulant unary embedding of a tridiagonal matrix as follows.
\begin{proposition}[Circulant unary embedding without penalty]\label{prop:circ_unary_ebd}
    Let $A\in\mathbb{C}^{N\times N}$ be a Hermitian, tridiagonal matrix. Let $\mathcal{S}$ be the subspace spanned by the set of circulant unary states $\{\ket{c_j} \mid j=1,\dots,N\}$. 
    For $j=1,\dots,N$, define $a_j$, $b_j$, and $c_j$ as follows:
    \begin{equation*}
        a_j = \begin{cases}
            1 & \text{if $N/2-1\leq j < N$,}\\
            0 & \text{otherwise,}
        \end{cases}\quad
        b_j = \begin{cases}
            1 & \text{if $1 < j \leq N/2$,}\\
            0 & \text{otherwise,}
        \end{cases}\quad
        c_j = \begin{cases}
            1 & \text{if $j \geq N/2$,}\\
            0 & \text{otherwise.}
        \end{cases}
    \end{equation*}
    Then the Hamiltonian
    \begin{equation}
        H_A = \sum_{j=1}^{N} \hat{n}_{j+1}^{(a_j)}\left(\Re(A_{j,j+1})) X_j + (-1)^{c_j} \Im(A_{j,j+1}) Y_j\right) \hat{n}_{j-1}^{(b_j)}
        + A_{j} \hat{n}_{j}^{(c_j)} \hat{n}_{j-1}^{(b_j)}
    \end{equation}
    (where indices are taken modulo $N/2$) satisfies $H_A \big|_{\mathcal{S}} = A$ and $\mathcal{S}$ is an invariant subspace of $H_A$.
\end{proposition}
Both \prop{unary_ebd} and \prop{circ_unary_ebd} are proved via straightforward calculation, as the operators involving $A_{j,j+1}$ only act nontrivially on the $j$-th and $(j+1)$-th codewords.

\section{Standard binary encodings of sparse matrices}\label{append:std_binary_enc}

\subsection{Polynomial functions}\label{append:std_binary_polynomials}
The standard binary encoding enables a compact representation of diagonal operators corresponding to polynomial functions of a real variable.
Here, we use $n$ qubits to represent a real-variable $x\in[0,1]$ and discretize the interval by $N=2^n$ grid points.
For a variable $y$ defined on a more general interval $[L,R]$, one may rescale and translate $x$ via $y=(R-L)x-L$.
\begin{proposition}\label{prop:std_binary_polynomial}
    Let $N=2^n$ for $n\geq 1$. Let $P = \diag(0, 1, \dots, N-1)/N$.
    Then the Pauli decomposition of $P$ is
    \begin{equation}
        P = \frac{1}{2N}\left((N-1)I - \sum_{j=1}^{n} 2^{j-1} Z_j \right)
    \end{equation}
\end{proposition}
\begin{proof}
    For any computational basis state $\ket{k}$, let $q_i\in\{0,1\}$ be the value of the $i$-th bit.
    Then $k=\sum_{i=1}^{n} 2^{i-1} q_{i}$.
    Then,
    \begin{equation}
        \bra{k} \sum_{j=1}^{n} 2^{j-1} Z_j \ket{k} = \sum_{j=1}^{n} 2^{j-1} (1 - 2q_j) = 2^{n} - 1 - 2k.
    \end{equation}
    Therefore,
    \begin{equation}
        \frac{k}{N} = \frac{1}{2N}\bra{k} \left(2^{n} - 1 - \sum_{j=1}^{n} 2^{j-1} Z_j \right)\ket{k} = \bra{k} P \ket{k}.
    \end{equation}
\end{proof}
\begin{remark}\label{rem:general_polynomials}
    It immediately follows that a general polynomial $\sum_{k=0}^{K} a_k x^k$ may be encoded by the operator $\sum_{k=0}^{K} a_k P^k$.
    As a consequence, the locality of the operator corresponds to the degree of the polynomial and the number of Pauli terms may be as large as $N$ for the standard binary code.
    In contrast, the one-hot code requires $N$ single-qubit operators, regardless of the degree $K$.
    An empirical comparison of the standard binary code and one-hot code for simulating diagonal operators is shown in \fig{diagonal_operators}.
\end{remark}
\begin{figure}[ht!]
    \centering
    \includegraphics[width=0.5\linewidth]{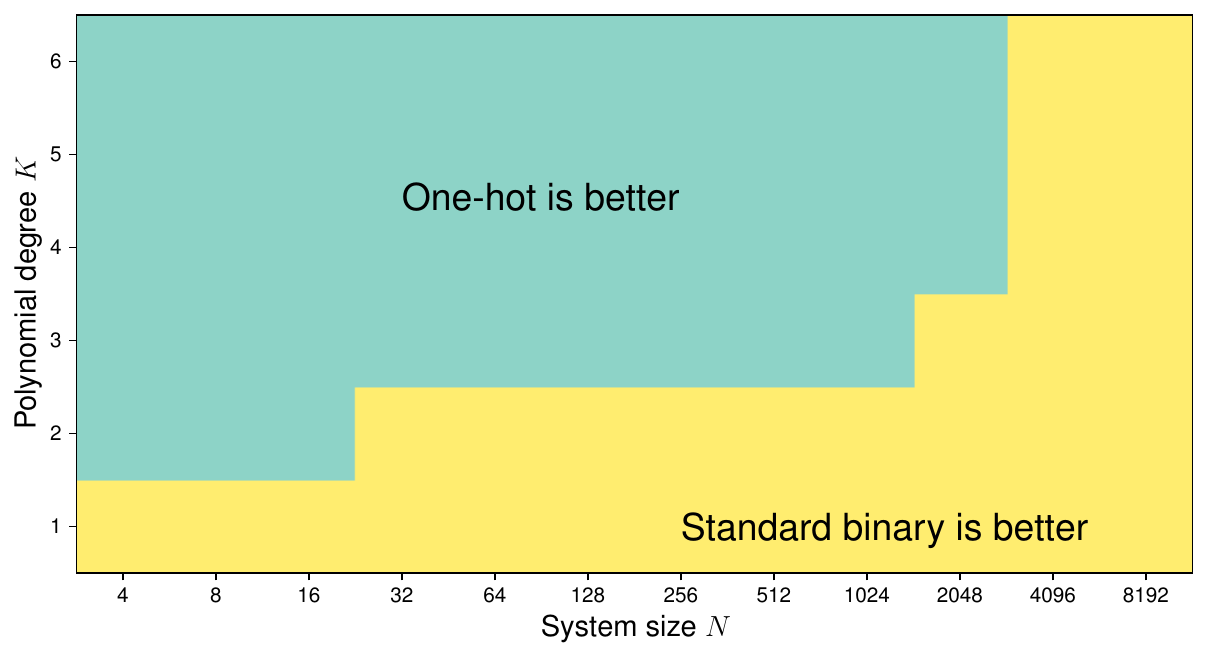}
    \caption{Landscape of the best encoding scheme for simulating a diagonal operator $P^K$ for varying system size $N$ and polynomial degree $K$, determined by the total number of 1- and 2-qubit gates.
    The yellow region shows the regime where standard binary uses fewer gates, while the teal region shows where the one-hot code uses fewer gates.
    For linear polynomials, the standard binary code outperforms the one-hot code for system size up to $N=256$.
    For higher degree polynomials, the one-hot code tends to use fewer gates, unless the system size is very large.}
    \label{fig:diagonal_operators}
\end{figure}

In particular, the standard binary encoding allows for a highly efficient decomposition of diagonal matrices containing Fourier frequencies.
For consistency with the QFT, we let the first entry be the zero-frequency component, unlike in \prop{std_binary_polynomial}.

\begin{proposition}\label{prop:fourier_freq}
    Let $N = 2^{n}$ for $n\geq 1$.
    Let $H_d = \mathrm{diag}(0,1,\dots,N/2-1,-N/2,-N/2+1,\dots,-1)$. Then the Pauli decomposition of $H_d$ is
    \begin{equation}
        H_d = -\frac{1}{2}\left(I - 2^{n} Z_{n} + \sum_{j=1}^{n} 2^{j-1} Z_j \right).
    \end{equation}
\end{proposition}
\begin{proof}
    Similar to the proof of \prop{std_binary_polynomial}, we have
    \begin{equation}
        \bra{k} \sum_{j=1}^{n} 2^{j-1} Z_j \ket{k} = 2^{n} - 1 - 2k,
    \end{equation}
    where $\ket{k}$ denotes the standard binary encoding of $k$. Then by straightforward calculation,
    \begin{equation}
        \bra{k} H_d \ket{k} =
        \begin{cases}
            k &\qquad\text{if $k \leq N/2-1$,}\\
            k-N &\qquad\text{if $k \geq N/2$.}
        \end{cases}
    \end{equation}
\end{proof}

\subsection{Finite difference operators}\label{append:std_binary_finite_diff}
Here, we discuss the standard binary encoding of finite difference operators and how this approach compares to sparse encodings such as the one-hot or unary codes.
We consider tridiagonal matrices of the form
\begin{equation}\label{eqn:finite_diff_op}
    A = \sum_{j=0}^{N-1} a_j \ket{j+1}\bra{j} + a_j^{*}\ket{j}\bra{j+1}, 
\end{equation}
which includes discretizations of both first and second-order differential operators, with or without periodic boundary conditions.

In the special case of homogeneous entries (i.e. $a_j=a$ for all $j$), circuit implementations for simulating \eqn{finite_diff_op} were proposed in~\cite{sato2024hamiltonian}.
These circuit constructions were subsequently utilized for simulating non-unitary dynamics via Schr\"odingerization in~\cite{hu2024quantum}.
Strictly speaking, this approach represents the finite difference operators using the standard binary encoding.
In general, the standard binary encoding for an $N$-point finite difference operator requires $\mathcal{O}(N)$ of Pauli terms despite using only $\mathcal{O}(\log N)$ qubits (see \cite{arseniev2024tridiagonal} for the Pauli decomposition of general tridiagonal matrices).
However, \cite{sato2024hamiltonian} circumvents the naive implementation by decomposing and simulating each term in the operator using the so-called Bell basis rather than the Pauli basis.
This approach offers an efficient implementation with only $\mathcal{O}(\log N)$ space and $\mathcal{O}(\poly\log N)$ gates per Trotter step, but suffers from two major disadvantages.

\begin{itemize}
    \item 
    First, the constructed circuits are only applicable to tridiagonal matrices with homogeneous entries.
    When considering differential equations with variable coefficients, the entries are generally inhomogeneous and must be represented by $\mathcal{O}(N)$ parameters.
    In these cases, dense encodings such as the standard binary code may save space, but cannot achieve $\mathcal{O}(\poly\log N)$ gates per Trotter step without further leveraging the problem structure.
    
    In the case of inhomogeneous entries, it is possible to simulate $a_j\ket{j}\bra{j+1}+a_j^{*}\ket{j}\bra{j}$ for each $j$ by applying increment and decrement operators (see for example~\cite{douglas2009efficient}) and simulating $(\ket{0}\bra{1})^{\otimes n} + (\ket{1}\bra{0})^{\otimes n}$ via the circuits from~\cite{sato2024hamiltonian}.
    Using this approach, the Trotter error is the same as one-hot or unary encodings.
    However, since this must be done for all $j=1,\dots,N$, the overall gate complexity is $\mathcal{O}(N \cdot \poly\log N)=\Tilde{\mathcal{O}}(N)$.
    This is worse than sparse encodings such as one-hot or unary only by logarithmic factors.
    \item 
    Second, the standard binary code requires highly non-local operators.
    As a result, the circuit implementation of these finite difference operators requires decomposing multi-controlled rotation operators and may lead to high overhead in practice.
    The actual costs of such decompositions depends largely on the device architecture.
\end{itemize}
It is also worth noting that even in the case of homogeneous entries, the use of the standard binary code for simulating partial differential equations is not exponentially more efficient than sparse encodings such as the one-hot or unary codes.
The use of finite differences leads to a factor of $1/h=\mathcal{O}(N)$ in the Hamiltonian, so the Trotter number scales at least linearly in $N$.
At best, the standard binary code achieves a gate complexity that is only a linear factor of $N$ better than the one-hot or unary codes when simulating finite difference operators.

\section{Quantum simulation of non-unitary dynamics}\label{append:nonunitary_dynamics}
We review two complementary methods for solving \eqn{linear_ode_time_indep} in the general case when $A(t)$ is not necessarily anti-Hermitian.
This includes linear combination of Hamiltonian simulation (LCHS)~\cite{an2023linear, an2023quantum} and Schr\"odingerization~\cite{jin2023quantum}, which are distinct yet closely related methods for mapping non-unitary dynamics to unitary dynamics.
We apply each of these methods in conjunction with Hamiltonian embedding to obtain resource-efficient algorithms for simulating non-unitary dynamics. 

\subsection{Linear combination of Hamiltonian simulation}
The original work on LCHS~\cite{an2023linear} establishes the following identity:
\begin{equation}\label{eqn:lchs_formula}
    \mathcal{T}e^{\int_0^{t} A(s)\,\d s}
    = \int_{\mathbb{R}} \frac{1}{\pi{(1+k^2)}}\mathcal{T}e^{i\int_0^{t} k H_1(s) + H_2(s)\,\d s}\d k.
\end{equation}
where $H_1(t)=\frac{A(t) + A(t)^\dagger}{2}$ and $H_2(t)=\frac{A(t) - A(t)^\dagger}{2i}$, assuming $H_1(t) \preccurlyeq 0$ for stability.

In~\cite{an2023quantum}, this formula was generalized by replacing the integrand $\frac{1}{\pi(1+k^2)}$ by a large family of functions to obtain a quantum algorithm with near-optimal dependence on all parameters for solving \eqn{linear_ode_time_indep}.
The general formula is given by
\begin{equation}\label{eqn:lchs_formula_general}
    \mathcal{T}e^{\int_0^{t} A(s)\,\d s}
    = \int_{\mathbb{R}} \frac{f(k)}{1-ik}\mathcal{T}e^{i\int_0^{t} k H_1(s) + H_2(s)\,\d s}\d k,
\end{equation}
where the kernel function $f(k)$ satisfies mild analyticity, decay, and normalization conditions~\cite[Theorem 5]{an2023quantum}.
The integration in \eqn{lchs_formula_general} is truncated to a finite interval $[-R,R]$ and approximated by a finite sum:
\begin{equation}
\begin{split}
    \mathcal{T}e^{\int_0^{t} A(s)\,\d s}
    &\approx \int_{-R}^{R} \frac{f(k)}{1-ik}\mathcal{T}e^{i\int_0^{t} k H_1(s) + H_2(s)\,\d s}\d k\\
    &\approx \sum_{j=0}^{M-1} c_j U_{j}(t)\label{eqn:lchs_approx},
\end{split}
\end{equation}
where $U_j(t) = \mathcal{T}e^{i\int_0^{t} k_j H_1(s) + H_2(s)\,\d s}$.
The weights $c_j$'s and quadrature points $k_j$'s depend on the chosen numerical integration scheme.

An observable $u(t)^\dagger O u(t)$ can be approximated by
\begin{equation}
    u(t)^\dagger O u(t) \approx \sum_{j_1,j_2} c_{j_1}^{*} c_{j_2} u(t)^\dagger U_{j_1}^\dagger(t) O U_{j_2}(t) u(t)\label{eqn:lchs_obs}.
\end{equation}
This observable can be computed by Monte Carlo sampling pairs $(j_1,j_2)$ with probability proportional to $c_{j_1}^{*}c_{j_2}$ and evaluating each correlation function $u(0)^\dagger U_{j_1}^\dagger(t) O U_{j_2}(t) u(0)$ on a quantum computer via the Hadamard test~\cite{tong2021fast}.

The choice of the truncation parameter $R$ required to approximate the integral in \eqn{lchs_formula_general} with precision $\epsilon$ depends on the choice of $f(k)$.
The original LCHS formula \eqn{lchs_formula}, which is obtained from setting $f(k)=\frac{1}{\pi(1+ik)}$, requires $K=\mathcal{O}(\frac{1}{\epsilon})$.
Using the improved kernel function $f(z)=\frac{1}{2\pi e^{-2^\beta} e^{(1+iz)^{\beta}}}$ for $\beta\in(0,1)$, the cutoff can be chosen as $K=\mathcal{O}((\log(1/\epsilon))^{1/\beta}$.

We summarize the steps for implementing the hybrid LCHS algorithm in \algo{lchs}.
\begin{center}
\begin{algorithm}[ht!]
\caption{Hybrid LCHS}\label{algo:lchs}
\KwData{Hamiltonian embeddings $\widetilde{H}_1$, $\widetilde{H}_2$, $\widetilde{O}$, encoded initial state $\widetilde{u}_0$, evolution time $t$, number of samples $N_{\text{samples}}$.}
\KwResult{Estimate of observable $u(T)^\dagger O u(T)$.}
\For{$i=1,\dots,N_{\text{samples}}$}{
Randomly sample pairs $(j_1,j_2)$ with probability proportional to $c_{j_1}^{*} c_{j_2}$, where $c_{j}$'s are weights as in \eqn{lchs_approx}. \\
Apply the Hadamard test to estimate the correlation function $\widetilde{u}(0)^\dagger U_{j_1}^\dagger(T) \widetilde{O} U_{j_2}(T) \widetilde{u}(0)$, where $U_{j_1}(T) = e^{it(k_{j_1} \widetilde{H}_1 + \widetilde{H}_2)}$ and $U_{j_2}(T) = e^{it(k_{j_2} \widetilde{H}_1 + \widetilde{H}_2)}$.
}
Estimate $u(T)^\dagger Ou(T)=\widetilde{u}(T)^\dagger \widetilde{O} \widetilde{u}(T)$ by averaging the values according to \eqn{lchs_obs}.
\end{algorithm}
\end{center}

\subsection{Schr\"odingerization}
The recently proposed method known as Schr\"odingerization~\cite{jin2023quantum} maps the linear ODE in \eqn{linear_ode_time_indep} to the Schr\"odinger equation, for which the dynamics are unitary.
This procedure makes use of the warped phase transformation, in which one defines an auxiliary variable $p$ and the mapping $v(t,p)=e^{-p} u(t)$ for $p>0$.
Rather than directly solving \eqn{linear_ode_time_indep}, one instead solves the linear convection equation
\begin{equation}
\begin{split}
    \frac{\partial}{\partial t} v(t,p) &= -H_1(t) \frac{\partial}{\partial p} v(t,p) + i H_2(t) v(t,p)\\
    v(0,p) &= e^{-|p|} u_0,\label{eqn:warped_phase}
\end{split}
\end{equation}
where $H_1(t)$ and $H_2(t)$ are the real and imaginary parts of $A(t)$ given by the Cartesian decomposition $A(t)=H_1(t) + i H_2(t)$.

The linear operator in \eqn{warped_phase} is anti-self-adjoint. 
As a result, \eqn{warped_phase} is a special case of the Schr\"odinger equation and can be solved using standard Hamiltonian simulation techniques.
Then the solution to \eqn{linear_ode_time_indep} can be recovered by integrating out the auxiliary variable:
\begin{equation}
    u(T)=\int_{0}^{\infty} v(T,p)\,\d p.\label{eqn:schrodingerization_recover}
\end{equation}
In practice, an observable of the form $u(T)^\dagger O u(T)$ is computed by post-selecting measurements for which $p>0$ and computing the expectation value of $O$ on the system register normally.

We summarize the steps for implementing Schr\"odingerization in \algo{schrodingerization}.
\begin{center}
\begin{algorithm}[ht!]
\caption{Schr\"odingerization}\label{algo:schrodingerization}
\KwData{Hamiltonian embeddings $\widetilde{H}_1$, $\widetilde{H}_2$, $\widetilde{O}$, encoded initial state $\widetilde{u}_0$, evolution time $T$, number of auxiliary qubits $n_p$.}
\KwResult{Estimate of observable $u(T)^\dagger O u(T)$.}
Using an $n_p$-qubit register to represent $p$, prepare the initial state $\ket{\Psi_{\text{Laplace}}}$ with amplitudes proportional to the Laplace distribution. \\
Apply the QFT to the $p$-register. \\
Simulate $\widetilde{H}_1 \otimes H_{\mathcal{F}} - \widetilde{H}_2 \otimes I$ for time $T$, where $H_{\mathcal{F}}$ is the diagonal matrix of Fourier frequencies. \\
Apply the inverse QFT to the $p$-register. \\
Estimate $u(T)^\dagger O u(T)$ via observable measurements and post-selection.
\end{algorithm}
\end{center}

\subsection{Proof of Schr\"odingerization}\label{append:schrodingerization}

LCHS and Schr\"odingerization accomplish exactly the same task of solving \eqn{linear_ode_time_indep}.
In fact, we remark that LCHS and Schr\"odingerization are mathematically equivalent, as both are based upon the LCHS formula \eqn{lchs_formula}.
\begin{theorem}\label{thm:schrodingerization_proof}
    Let $u(t)$ be the solution to the linear ODE in \eqn{linear_ode_time_indep} and let $v(t,p)$ be the solution to the PDE \eqn{warped_phase}.
    Then $u(t)=\int_{0}^{\infty} v(t,p)\,\d p$.
\end{theorem}
\begin{proof}
    Applying the Fourier transform to $v(t,p)$ with respect to $p$ yields
    \begin{equation}
        \hat{v}(t,k) = \frac{1}{\sqrt{2\pi}} \int_{\mathbb{R}} e^{ipk} v(t,p)\, \d p.
    \end{equation}
    Then \eqn{warped_phase} becomes
    \begin{align}\label{eqn:schrodinger}
        \frac{\partial}{\partial t} \hat{v}(t,k)
        &= i\left(k H_1(t) + H_2(t) \right)\hat{v}(t,k),
    \end{align}
    with initial condition $\hat{v}(0,k) = \frac{\sqrt{2/\pi}}{k^2 + 1} u_0$.
    The Schr\"odinger equation \eqn{schrodinger} has solution
    \begin{equation}
        \hat{v}(x,k) = \mathcal{T}e^{i\int_{0}^{t} \left(k H_1(s) + H_2(s)\right) \,\d s}\frac{\sqrt{2/\pi}}{k^2+1} u_0.
    \end{equation}
    The solution to the original PDE $v(t,p)$ is recovered by taking the inverse Fourier transform:
    \begin{align}
        v(t,p)
        &= \frac{1}{\sqrt{2\pi}} \int_{\mathbb{R}} e^{-i p k} \hat{v}(t,k) \,\d k\, u_0\\
        &= \frac{1}{\sqrt{2\pi}} \int_{\mathbb{R}} e^{-i p k} \mathcal{T}e^{i\int_{0}^{t} \left(k H_1(s) + H_2(s)\right)\,\d s}\frac{\sqrt{2/\pi}}{k^2+1}\,\d k\, u_0\\
        &=  \int_{\mathbb{R}} \mathcal{T}e^{i\int_{0}^{t}(k(H_1(s) - p/t) + H_2(s))\,\d s} \frac{1}{\pi(k^2+1)} \,\d k\, u_0 .
    \end{align}
    Since $H_1\preccurlyeq 0$,  $H_1-p/t$ is negative if $p\geq 0$.
    Using \cite[Theorem~1]{an2023linear}, we have
    \begin{equation}
        v(t,p) = e^{-p} \mathcal{T} e^{\int_{0}^{t} A(s)\,\d s} u_0 = e^{-p} u(t), \qquad\qquad p\geq 0.
    \end{equation}
    Thus, we have $u(t) = \int_{0}^{\infty} v(t,p)\, \d p$ since $\int_{0}^{\infty} e^{-p} \d p= 1$.
\end{proof}
\begin{remark}
    Schr\"odingerization can be generalized by choosing the initial state to be the Fourier transform of any valid kernel function in \cite{an2023quantum}.
    The proof of correctness is nearly identical, with the only difference being the initial function.
    Compared to the original integrand $\frac{1}{\pi(1+\eta^2)}$, these kernel functions have better decay properties that enable a smaller truncation interval in the asymptotic regime.
    In \append{schrodingerization_implementation}, we show a straightforward protocol for preparing the Laplace distribution.
    However, our proof of \thm{complexity_analysis_formal} depends on the use of the improved kernel functions and assumes an efficient protocol for preparing them as a quantum state.
\end{remark}

\subsection{Implementation of Schr\"odingerization}
\label{append:schrodingerization_implementation}
In this section, we discuss details of the end-to-end implementation of Schr\"odingerization.
The overall implementation consists of three components: (1) initial state preparation, (2) quantum simulation of the linear convection equation in \eqn{warped_phase}, as well as (3) measurement of observables.

The overall quantum circuit is shown in \fig{schrodingerization_circuit}.
We use two quantum registers: one for $p$ and one for $u$.
We use the standard binary code with $n_p$ qubits for $p$, while the encoding for $u$ largely depends on the target application.
The quantum state proportional to the discretized solution to \eqn{warped_phase} is denoted as $\ket{v(t,p)}$.

To represent the $p$ variable in a quantum register, we first truncate the real space to the interval $[-R,R]$.
Note that the endpoint $R$ should be chosen to scale as $\mathcal{O}(\|H_1\|t)$ since the wave function travels in the negative $p$ direction at a rate proportional to $\|H_1\|t$.
Next, we discretize with $N_p=2^{n_p}$ grid points $p_0,\dots,p_{N_p-1}$ such that $p_k=-R+k\Delta p$, where $\Delta p=\frac{2R}{N_p-1}$.
For $j=0,\dots,N_p-1$, we use $\ket{j}$ to represent the point $p_j$, where $\ket{j}$ denotes the standard binary representation of $j$.

\paragraph{Initial state preparation.}
The first step of Schr\"odingerization is the task of initial state preparation of the Laplace distribution.
As the initial condition of \eqn{warped_phase} is $v(0,p)=e^{-|p|} u_0$, we must prepare an $n_p$-qubit quantum state whose amplitudes are proportional to the Laplace distribution $f(p)=e^{-|p|}$, $p\in\mathbb{R}$. 
Here, we resolve this initialization issue, which to our knowledge has not been fully addressed in previous literature.

In what follows, we omit the implicit normalization factor for simplicity.
This normalization factor does not appear when computing observables.
Formally, the goal is to approximate the quantum state $\ket{\psi_{\text{Laplace}}} = \sum_{j=0}^{N_p-1} e^{-|p_j|}\ket{j}$.

To construct a circuit for the Laplace distribution, we first consider the right half of the real line, i.e. the exponential distribution $e^{-p}$ for $p\geq 0$.
We start by showing how one may prepare an $n$-qubit state $\ket{\psi_{\text{Exp},n}} = \sum_{j=0}^{2^{n}-1} e^{-|q_j|} \ket{j}$, where $q_j=jh$ and $h=\frac{R}{2^n}$.

A key observation is that $\ket{\psi_{\text{Exp},n}}$ is a product state:
\begin{equation}
    \ket{\psi_{\text{Exp},n}} 
    = \bigotimes_{i=0}^{n-1} \left(\ket{0} + e^{-2^{n-1-i} h}\ket{1}\right).
\end{equation}

For each $i=0,\dots,n-1$, the single-qubit state $\left(\ket{0} + e^{-2^{n-1-i} h}\ket{1}\right)$ is readily prepared using a single Pauli-$y$ rotation with angle $\theta_i=2 \arccos(\frac{1}{\sqrt{1 + e^{-2 (2^{n-1-i}) h)}}})$.
Thus $\ket{\psi_{\text{Exp},n}}$ can be prepared by a single layer of Pauli-$y$ rotations.

To prepare the Laplace distribution on $n_p$ qubits, we start with an $(n_p-1)$-qubit state $\ket{\psi_{\text{Exp},n_p-1}}$ for the exponential distribution.
We introduce an additional qubit as the most significant bit.
Let $q_0,\dots,q_{n_p-1}$ denote the qubits where $q_{n_p-1}$ is the newly introduced most significant bit.
We symmetrize the distribution with a Hadamard gate on $q_{n_p-1}$, followed by a sequence of $n_p-1$ CNOT gates each controlled on $q_{n_p-1}$ and targeting $q_i$ for $i=0,\dots,n_p-2$, followed by an $X$ gate on $q_{n_p-1}$.
The full state preparation circuit is shown in \fig{laplace_dist_circuit}.
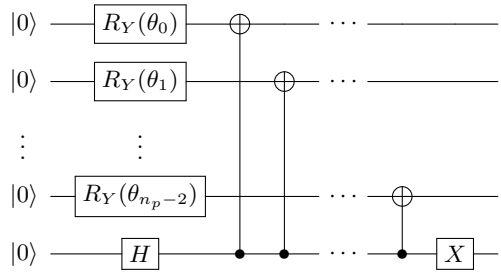
\begin{figure}[ht]
    \centering
    \quad\quad\qquad
    \Qcircuit @C=1em @R=.7em @!R {
        \lstick{\ket{0}} & \gate{R_Y(\theta_0)} & \targ & \qw & \qw & \cdots & & \qw & \qw & \qw \\
        \lstick{\ket{0}} & \gate{R_Y(\theta_1)} & \qw & \targ & \qw & \cdots & & \qw & \qw & \qw \\
        \lstick{\vdots\,\,\,} & \vdots & & & & & \\
        \lstick{\ket{0}} & \gate{R_Y(\theta_{n_p-2})} & \qw & \qw & \qw & \cdots & & \targ & \qw & \qw \\
        \lstick{\ket{0}} & \gate{H} & \ctrl{-4} & \ctrl{-3} & \qw & \cdots & & \ctrl{-1} & \gate{X} & \qw
    }
    \caption{Quantum circuit for preparing the Laplace distribution.}
    \label{fig:laplace_dist_circuit}
\end{figure}

\paragraph{Hamiltonian simulation.}
Next, we discuss the simulation of the linear convection equation given by \eqn{warped_phase}.
When mapped to the Schr\"odinger equation, the Hamiltonian to simulate is given as $-iH_1(t)\partial_p -H_2(t)$ (before discretization of $p$).
To implement $-i\partial_p$, we enforce periodic boundary conditions and apply the Fourier spectral method to the auxiliary variable $p$.
By \prop{fourier_freq}, the matrix of Fourier frequencies is given by a sum of 1-local Pauli-$z$ operators:
\begin{equation}
    H_\mathcal{F} = -\frac{\pi}{2R}\left(I - 2^{n_p} Z_{n_p} + \sum_{j=1}^{n_p} 2^{j-1} Z_j \right).
\end{equation}
Upon applying Hamiltonian embedding, $H_1$ and $H_2$ are mapped to $\widetilde{H}_1$ and $\widetilde{H}_2$, respectively.
Then, the overall embedding Hamiltonian to simulate is given by
\begin{equation}\label{eqn:schrodingerization_hamiltonian}
    \widetilde{H}_{S} = \widetilde{H}_1 \otimes H_{\mathcal{F}} - \widetilde{H}_2 \otimes I,
\end{equation}
which can be simulated using standard methods for sparse Hamiltonian simulation (e.g., product formulas).
Finally, an application of the inverse QFT is applied to the $p$-register to obtain
\begin{equation}
    \ket{\psi(t)} = (I \otimes \mathcal{F}^{-1}) e^{-i t H_{S}} (I \otimes \mathcal{F})(\ket{u_0}\otimes\ket{\psi_{\text{Laplace}}}),
\end{equation}
where $\mathcal{F}$ denotes the quantum Fourier transform.
\paragraph{Measurement of observables.}
The final step of Schr\"odingerization is to recover the solution $u(t)$ by integrating out the auxiliary variable according to \eqn{schrodingerization_recover}.
Since $v(t,p)=e^{-p}u(t)$ for $p>0$, we have
\begin{equation}
    u(t)^\dagger O u(t) = 2\int_{0}^{\infty} v(t,p)^\dagger O v(t,p)\,\d p.
\end{equation}
Thus, observables of the form $u(t)^\dagger O u(t)$ may be computed by measuring the $p$-register and post-selecting for values of $p>0$, obtaining the unnormalized state
\begin{equation}
    \rho_{u} = \sum_{k=N_p/2}^{N_p-1} (I \otimes \bra{k}) \ket{\psi(t)}\bra{\psi(t)} (I \otimes \ket{k}).
\end{equation}
Note that this is not a partial trace operation and $\rho_{u}$ is not a density operator, since $\Tr\rho_{u} < 1$.
Finally, the observable is estimated by
\begin{equation}
    u(t)^\dagger O u(t) \approx 2 \Tr(O\rho_{u}).
\end{equation}

\section{Quantum ODE solvers using Hamiltonian embedding}

\subsection{Hamiltonian simulation via product formulas and Richardson extrapolation}\label{append:richardson_extrapolation}
Our description of LCHS and Schr\"odingerization thus far have not specified the method for performing Hamiltonian simulation.
In this work, we consider the use of product formulas combined with Richardson extrapolation, which enables us to achieve near-optimal complexity for simulating ODEs.

Given a Hamiltonian $H=\sum_{\gamma=1}^{\Gamma} H_\gamma$, an $\Upsilon$-stage product formula $\mathcal{P}$ is given by~\cite{childs2021theory}
\begin{equation}\label{eqn:staged_product_formula}
    \mathcal{P}(t) \coloneqq \prod_{\upsilon=1}^{\Upsilon} \prod_{\gamma=1}^{\Gamma} e^{-ita_{(\upsilon,\gamma)} H_{\pi_\upsilon(\gamma)}},
\end{equation}
where $\Upsilon$ is the number of stages.
For a $2k$-th order Trotter-Suzuki formula, the number of stages is $\Upsilon = 2\cdot 5^{k-1}$.
For a given product formula $\mathcal{P}$, the symmetry class $\sigma$ is 2 if $\mathcal{P}$ is symmetric (i.e. $\mathcal{P}(-t)=\mathcal{P}^{-1}(t)$), or 1 otherwise.
Throughout this section, we also let $a_{\mathrm{max}}=\max_{\upsilon,\gamma}|a_{(\upsilon,\gamma)}|$.

We use the notation $[X_1 X_2 \dots X_n]$ to denote the right-nested commutators:
\begin{equation}
    [X_1 X_2 \dots X_n ] \coloneqq \comm{X_1}{\comm{X_2}{\comm{\dots}{\comm{X_{n-1}}{X_{n}}}}},
\end{equation}
and define
\begin{equation}
    \alpha_{\mathrm{comm}}^{(j)} \coloneqq \sum_{\gamma_1,\dots,\gamma_j=1}^{\Gamma} \|[H_{\gamma_1} H_{\gamma_2} \dots H_{\gamma_j}]\|.
\end{equation}
It has been shown that the error of a $p$-th order staged product formula has commutator scaling~\cite{childs2021theory}:
\begin{equation}
    \|e^{-tiH} - \mathcal{P}(t)\| = \mathcal{O}(\alpha_{\mathrm{comm}}^{(p+1)} t^{p+1})
\end{equation}
for small $t$.

The direct application of product formulas to simulate long-time evolution with accuracy $\epsilon$ results in runtime the scales polynomially in $1/\epsilon$.
In \cite{watson2025exponentially}, it was shown that the dependence on $\epsilon$ can be improved to $\log(1/\epsilon)$ using classical post-processing in the form of Richardson extrapolation.

\begin{proposition}[Richardson extrapolation, \cite{watson2025exponentially}]
    Let $f:[-1,1]\to\mathbb{R}$ have the series expansion
    \begin{equation}
        f(s)=f(0)+\sum_{j=1}^{\infty} c_j s^j.
    \end{equation}
    Let $s_0>0$ and choose extrapolation nodes $s_i=s_0/r_i$ for integers $r_1,\dots,r_m$.
    Then by choosing coefficients $b_k$ appropriately, the Richardson extrapolation
    \begin{equation}\label{eqn:richardson_extrapolation}
        F^{(m)}(s_0) = \sum_{k=1}^{m} b_k f(s_0/r_k)
    \end{equation}
    satisfies
    \begin{equation}
        |F^{(m)}(s_0) - f(0)|=\mathcal{O}(s_0^{m+1}).
    \end{equation}
\end{proposition}
\begin{remark}
    The conditioning of the extrapolation depends on the choice of extrapolation nodes, which affects the 1-norm of the coefficients $b_j$. In particular, one may choose 
    \begin{equation}
        r_k = r_{\mathrm{scale}} \left\lceil\frac{\sqrt{8} m}{\pi\sin(\pi(2k-1)/8m)}\right\rceil,
    \end{equation}
    which gives $\|b\|_1=\mathcal{O}(\log m)$~\cite{watson2025exponentially}.
\end{remark}

In the context of Hamiltonian simulation, we define 
\begin{equation}
    f(s)=\expval{O_{s}} \coloneqq \Tr \left( O \mathcal{P}(sT)^{1/s}\rho_0 (\mathcal{P}(sT)^{1/s})^\dagger\right),
\end{equation}
where $s>0$ represents the inverse Trotter number.
We are interested in the zero step size limit, for which $\lim_{s\to 0} \expval{O_s} = \Tr \left( O e^{-iTH}\rho_0 e^{iTH}\right) \eqqcolon \expval{O_0}$.
For an $m$-term Richardson extrapolation, we evaluate the sequence of estimates $f(s_1),\dots,f(s_m)$ and use these values to estimate the expectation of the time-evolved observable $\expval{O_0}$. 

\begin{lemma}[Trotter steps for Richardson extrapolation, {\cite[Corollary 4]{watson2025exponentially}}]\label{lem:extrapolation}
    Let $H=\sum_{\gamma=1}^{\Gamma} H_\gamma$ be a Hamiltonian, $O$ be an observable, $T>0$ be the evolution time, and $\rho_0$ be an initial state.
    Let $\mathcal{P}$ be a symmetric $p$-th order $\Upsilon$-stage product formula with symmetry class $\sigma$.
    Consider an $m$-term Richardson extrapolation \eqn{richardson_extrapolation} for estimating $\expval{O_0}\coloneqq\Tr \left( O e^{-iTH}\rho_0 e^{iTH}\right)$ via the estimates $\expval{O_{s_j}}$ for $j=1,\dots,m$.
    Assume $s_1$ is chosen such that $a_{\mathrm{max}}\Upsilon s_1 \lambda T < 1/2$ and $r_{\mathrm{scale}}$ is chosen large enough such that 
    \begin{equation}
        r_1 \geq (a_{\mathrm{max}} \Upsilon \lambda T)^{1+\frac{1}{\sigma m}\lceil\frac{\sigma m}{p}\rceil}\left(\frac{4\|b\|_1}{\epsilon}\right)^{\frac{1}{\sigma m}}.
    \end{equation}
    Then to estimate $\expval{O_0}$ with error $\epsilon$, the maximum number of Trotter steps in a single circuit is 
    \begin{equation}
        \mathcal{O}\left((a_{\mathrm{max}} \Upsilon \lambda T)^{(1+1/p)} \log(1/\epsilon)\right),
    \end{equation}
    where $a_{\mathrm{max}} = \max_{\upsilon,\gamma} a_{(\upsilon,\gamma)}$, and
    \begin{equation}
        \lambda \coloneqq \sup_{\substack{j\in\sigma \mathbb{Z}_{+}\geq \sigma m \\ 1\leq l \leq K}} \left(\sum_{\substack{j_1,\dots,j_l\in\sigma \mathbb{Z}_{+}\geq p \\ j_1 + \dots + j_l = j}} \prod_{\kappa=1}^{l} 2\frac{\alpha_{\mathrm{comm}}^{(j_\kappa+1)}}{(j_\kappa+1)^2}\right)^{1/(j+l)}.\label{eqn:lamb_comm}
    \end{equation}
\end{lemma}

\begin{lemma}[Upper bound on $\lambda$ for $k$-local Hamiltonians{~\cite[Remark 3.2]{chakraborty2025quantum}}]
    Let $H=\sum_{j_1,\dots,j_k} H_{j_1,\dots,j_k}$ be a $k$-local Hamiltonian.
    Let
    \begin{equation}
        \vvvert H \vvvert_1 = \max_{\ell} \left\{\max_{j_{\ell}} \sum_{j_1,\dots,j_{\ell-1},j_{\ell+1},j_{k}} \|H_{j_1,\dots,j_k}\|\right\}
    \end{equation}
    denote the induced 1-norm.
    Then, \eqn{lamb_comm} is upper bounded by
    \begin{equation}
        \lambda = \mathcal{O}\left(\left(\frac{\|H\|_1}{\vvvert H\vvvert_1}\right)^{1/p}\vvvert H\vvvert_1\right).
    \end{equation}
\end{lemma}

\subsection{Complexity analysis for \texorpdfstring{\algo{pde_solver}}{Algorithm \ref*{algo:pde_solver}}}\label{append:complexity_analysis}
We rigorously analyze the computational complexity for \algo{pde_solver} when using either LCHS or Schrodingerization to reduce the problem to Hamiltonian simulation.
For the hybrid LCHS algorithm, each application of the Hadamard test can be viewed as simulating the Hamiltonian
\begin{equation}\label{eqn:lchs_ham}
    H_{\mathrm{LCHS}}=-\left(k_{j_1} \widetilde{H}_1 \otimes \ket{0}\bra{0} + k_{j_2} \widetilde{H}_1 \otimes \ket{1}\bra{1} + \widetilde{H}_2 \otimes I\right),
\end{equation}
where $k_{j_1}$ and $k_{j_2}$ are chosen randomly between $-R$ and $R$.
Similarly, Schr\"odingerization requires the simulation of the Hamiltonian
\begin{equation}\label{eqn:schro_ham}
    H_S
    = \widetilde{H}_1 \otimes H_{\mathcal{F}} - \widetilde{H}_2 \otimes I
\end{equation}
where $H_{\mathcal{F}}$ is a diagonal matrix of Fourier frequencies. Note that $\|H_{\mathcal{F}}\|_1 =\mathcal{O}(1/\Delta p)$ for step size $\Delta p$.

\begin{lemma}[Upper bound on $\lambda$ for LCHS and Schrodingerization]\label{lem:comm_upper_bound}
    Let $H_1$ and $H_2$ be Hamiltonians with Hamiltonian embeddings given by $\widetilde{H}_1=\sum_{\ell=1}^{L} \alpha_{\ell}^{(1)} P_\ell$ and $\widetilde{H}_2=\sum_{\ell=1}^{L} \alpha_{\ell}^{(2)} P_\ell$. Assume $\widetilde{H}_1$ and $\widetilde{H}_2$ are $k$-local, i.e., all $P_{\ell}$ act non-trivially on at most $k$ qubits.
    Let $H_{\mathrm{LCHS}}$ and $H_{S}$ be defined as in \eqn{lchs_ham} and \eqn{schro_ham}, respectively.
    Then when using a staged product formula $\mathcal{P}$ to simulate either $H_{\mathrm{LCHS}}$ or $H_S$, the commutator scaling prefactor $\lambda$ as defined in \eqn{lamb_comm} is given by
    \begin{equation}
        \lambda_{\mathrm{LCHS}} = \mathcal{O}\left((R\|\widetilde{H}_1\|_1 + \|\widetilde{H}_2\|_1)^{1/p} (R\vvvert \widetilde{H}_1 \vvvert_1 + \vvvert \widetilde{H}_2 \vvvert_1)^{1-1/p}\right)
    \end{equation}
    for LCHS, and
    \begin{equation}
        \lambda_{S} = \mathcal{O}\left(({\Delta p}^{-1}\|\widetilde{H}_1\|_1 + \|\widetilde{H}_2\|_1)^{1/p} ({\Delta p}^{-1}\vvvert \widetilde{H}_1 \vvvert_1 + \vvvert \widetilde{H}_2 \vvvert_1)^{1-1/p}\right)
    \end{equation}
    for Schrodingerization.
\end{lemma}
\begin{proof}
    For LCHS,
    \begin{align}
        \|H_{\mathrm{LCHS}}\|_1 
        &= \left\|-\left(\left(\frac{k_{j_1} + k_{j_2}}{2}\right) (\widetilde{H}_1 \otimes Z) + \left(\frac{k_{j_1}-k_{j_2}}{2}\right)(\widetilde{H}_1 \otimes Z) + \widetilde{H}_2\otimes I\right)\right\|_1\\
        &\leq (R+1)\|\widetilde{H}_1\|_1 + \|\widetilde{H}_2\|\\
        &=\mathcal{O}\left(R\|\widetilde{H}_1\|_1 + \|\widetilde{H}_2\|_1\right).
    \end{align}
    For Schr\"odingerization,
    \begin{align}
        \|H_S\|_1 
        &= \left\|\widetilde{H}_1 \otimes H_{\mathcal{F}} - \widetilde{H}_2 \otimes I\right\|_1\\
        &\leq \|\widetilde{H}_1 \otimes H_F\|_1 + \|\widetilde{H}_2 \otimes I \|_1\\
        &= \mathcal{O}\left(\frac{1}{\Delta p} \|\widetilde{H}_1\|_1 + \|\widetilde{H}_2\|_1\right).
    \end{align}
    Similarly, we have $\vvvert H_{\mathrm{LCHS}}\vvvert_1 =\mathcal{O}\left(R\vvvert\widetilde{H}_1 \vvvert_1 + \vvvert\widetilde{H}_2\vvvert_1\right)$ and $\vvvert H_{S} \vvvert_1 = \mathcal{O}\left({\Delta p}^{-1} \vvvert \widetilde{H}_1 \vvvert_1 + \vvvert \widetilde{H}_2 \vvvert_1 \right)$ for the induced 1-norm.
\end{proof}

Finally, our formal complexity analysis is given by \thm{complexity_analysis_formal}.
\begin{theorem}\label{thm:complexity_analysis_formal}
    Let $H_1$ and $H_2$ be Hamiltonians with Hamiltonian embeddings given by $\widetilde{H}_1=\sum_{\ell=1}^{L} \alpha_{\ell}^{(1)} P_\ell$ and $\widetilde{H}_2=\sum_{\ell=1}^{L} \alpha_{\ell}^{(2)} P_\ell$.
    Assume $\widetilde{H}_1$ and $\widetilde{H}_2$ are both $k$-local.
    Suppose we use either LCHS or Schrodingerization, along with a $q$-th order staged product formula and Richard extrapolation to estimate $u(T)^\dagger O u(T)$ with relative error $\epsilon$.
    Then under the conditions of \lem{extrapolation}, the total number of 1- and 2-qubit gates per circuit is
    \begin{equation}
        \mathcal{O}\left(kL\Upsilon^{2+o(1)}(a_{\mathrm{max}}T)^{1+o(1)} \log(1/\epsilon)^{3+o(1)}\left(\|\widetilde{H}_1\|_1 + \|\widetilde{H}_2\|_1\right)^{o(1)}\left(\vvvert \widetilde{H}_1 \vvvert_1 + \vvvert\widetilde{H}_2 \vvvert_{1}\right)^{1-o(1)}\right).
    \end{equation}
    When using LCHS, the sample complexity (i.e., number of different circuits) required to estimate $u(T)^\dagger O u(T)$ with relative error $\epsilon$ with probability at least $1-\delta$ is $\mathcal{O}\left(\frac{1}{\epsilon^2}\log\left(\frac{1}{\delta}\right)\right)$.
\end{theorem}
\begin{proof}
    By \lem{extrapolation}, the maximum Trotter number required is $\mathcal{O}\left((a_{\mathrm{max}} \Upsilon \lambda T)^{(1+1/p)} \log(1/\epsilon)\right)$, where $\lambda$ is the commutator scaling prefactor for either $H_{\mathrm{LCHS}}$ or $H_S$.
    We then bound $\lambda$ using \lem{comm_upper_bound} for LCHS and Schr\"odingerization separately.
    
    For LCHS, we have $\mathcal{O}(R)=\mathcal{O}((\log(1/\epsilon)^{1/\beta})$ by choosing the kernel function $f(k)=\frac{1}{2\pi e^{-2^{\beta}}e^{(1+ik)^{\beta}}}$, where $\beta\in(0,1)$~\cite{an2023quantum}.
    For Schr\"odingerization, we have $\|H_{S}\| = \mathcal{O}({\Delta p}^{-1})$.
    The discretization error from spectral methods depends on the smoothness of the initial function $v(0,p)=\psi(p)u_0$, where we choose $\psi(p)$ to be the Fourier transform of $f(k)=\frac{1}{2\pi e^{-2^{\beta}}e^{(1+ik)^{\beta}}}$.
    Since $f(k)$ decays exponentially fast, ${\Delta p}^{-1}=\mathcal{O}(\log(1/\epsilon))$ suffices.
    For more detailed analysis of Schr\"odingerization, we refer to \cite[Section 4.1]{jin2025schrodingerization}.
    For generality, we bound the norm of $H_{\mathrm{LCHS}}$ and $H_{S}$ by $\mathcal{O}((\log(1/\epsilon))^{2})$.
    Thus,
    \begin{equation}
        \lambda = \mathcal{O}\left(\log(1/\epsilon)^{2} \left(\|\widetilde{H}_1\|_1 + \|\widetilde{H}_2\|_1\right)^{1/p}\left(\vvvert \widetilde{H}_1 \vvvert_1 + \vvvert\widetilde{H}_2 \vvvert_{1}\right)^{1-1/p}\right)
    \end{equation}
    For the gate complexity, an additional factor of $\mathcal{O}(kL\Upsilon)$ arises from exponentiating and decomposing each term, thereby yielding the claimed gate complexity.

    To estimate $u(T)^\dagger O u(T)$ with absolute error $\epsilon$, the sample complexity when using importance sampling is $\mathcal{O}\left(\frac{\|O\|^2}{\epsilon^2}\log\left(\frac{1}{\delta}\right)\right)$, which follows from the use of Hoeffding's inequality~\cite{an2023linear, an2023quantum}.
    For a relative error $\epsilon$ (i.e., absolute error $\|O\|\epsilon$), the sample complexity is simply $\mathcal{O}\left(\frac{1}{\epsilon^2}\log\left(\frac{1}{\delta}\right)\right)$.
\end{proof}

\section{Resource estimates for the advection equation and nonlinear scalar hyperbolic PDE.}\label{append:explicit_counts}

Here, we explicitly provide the empirical estimates of resources used to simulate PDEs in the main text.
For the advection equation in \sec{linear-model}, the resource estimates are presented below.

\begin{table}[H]
    \centering
    \begin{tabular}{|c|c|c|c|c|c|c|}
    \hline
    \diagbox{\textbf{Encoding/embedding}}{$N$} & 8 & 16 & 32 & 64 & 128 & 256 \\
    \hline
    Standard binary (Pauli basis) & 2436552 & 26297612 & 218866288 & 1578432842 & 10836348028 & 70499518052 \\
    \hline
    Standard binary (Bell basis) & 763908 & 7458924 & 43759424 & 225588544 & 1012550952 & 4414678912 \\
    \hline
    One-hot & 639484 & 1744244 & 4769716 & 13250180 & 36951880 & 103520984 \\
    \hline
    Unary & 10125952 & 40444376 & 130567302 & 366006134 & 1022829818 & 2857164016 \\
    \hline
    \end{tabular}
    \caption{Circuit depths for simulating the linear advection equation \eqn{2d_adv_eq} with $N$ grid points per dimension.}
    \label{tab:advection_circ_depth}
\end{table}

\begin{table}[H]
    \centering
    \begin{tabular}{|c|c|c|c|c|c|c|}
    \hline
    \diagbox{\textbf{Encoding/embedding}}{$N$} & 8 & 16 & 32 & 64 & 128 & 256 \\
    \hline
    Standard binary (Pauli basis) & 2629832 & 28275344 &  234049044 & 1687121012 & 11502332116 & 74721198352 \\
    \hline
    Standard binary (Bell basis) & 767438 & 7605024 & 44760384 & 229395044 & 1027650880 & 4466948512 \\
    \hline
    One-hot & 2825728 & 15414848 & 84305536 &  468399872 & 2612538112 & 14638120448\\
    \hline
    Unary & 11014808 & 57868288 & 312940628 & 1736912548 & 9658783388 & 53824469152 \\
    \hline
    \end{tabular}
    \caption{Two-qubit gates for simulating the linear advection equation \eqn{2d_adv_eq} with $N$ grid points per dimension.}
    \label{tab:advection_two_qubit_gates}
\end{table}

For the nonlinear scalar hyperbolic PDE in \sec{nonlinear-model}, the resource estimates are presented below.

\begin{table}[H]
    \centering
    \begin{tabular}{|c|c|c|c|c|c|c|}
    \hline
    \diagbox{\textbf{Encoding/embedding}}{$N_q$} & 8 & 16 & 32 & 64 & 128 & 256 \\
    \hline
    Standard binary (Pauli) & 3438 & 10058 & 25149 & 56791 & 111434 & 173254 \\
    \hline
    One-hot & 448 & 896 & 1792 & 3584 & 7168 & 14336 \\
    \hline
    Unary & 900 & 1348 & 2244 & 4036 & 7620 & 14788 \\
    \hline
    \end{tabular}
    \caption{Circuit depths for simulating the nonlinear scalar hyperbolic PDE \eqn{nonlinear_pde} with $N_q$ grid points for $q$.}
    \label{tab:nonlinear_circ_depth}
\end{table}

\begin{table}[H]
    \centering
    \begin{tabular}{|c|c|c|c|c|c|c|}
    \hline
    \diagbox{\textbf{Encoding/embedding}}{$N_q$} & 8 & 16 & 32 & 64 & 128 & 256 \\
    \hline
    Standard binary (Pauli) & 4030 & 12023 & 29496 & 66265 & 127828 & 196039 \\
    \hline
    One-hot & 4144 & 4816 & 6160 & 8848 & 14224 & 24976 \\
    \hline
    Unary & 5132 & 5804 & 7148 & 9836 & 15212 & 25964 \\
    \hline
    \end{tabular}
    \caption{Two-qubit gates for simulating the nonlinear scalar hyperbolic PDE \eqn{nonlinear_pde} with $N_q$ grid points for $q$.}
    \label{tab:nonlinear_two_qubit_gates}
\end{table}

\section{Comparison between Schr\"odingerization and LCHS}\label{append:lchs_comparison}

We provide a comparison between Schr\"odingerization and LCHS when used to simulate non-Hermitian dynamics.
As a simple example, we consider the simulation of the transverse-field Ising model (TFIM) with an imaginary longitudinal field:
\begin{equation}
    A = -iH = -\gamma_z\sum_{j=1}^{n}(I-Z_j) -i\left(J\sum_{j=1}^{n-1} Z_j Z_{j+1} + h\sum_{j=1}^{n} X_{j}\right).
\end{equation}

For Schr\"odingerization, we consider $N_{\text{S,shots}}$ shots per circuit, where $N_{\text{S,shots}}\in\{10^2, 10^3, 10^4, 10^5\}$.
The hybrid LCHS algorithm~\cite{an2023linear} involves Monte Carlo sampling which induces sampling error in addition to shot noise.
Here, we consider the combined standard error due to both sources of randomness.
In our simulation, we consider $N_{\text{LCHS,samples}}$ randomly sampled circuits and $N_{\text{LCHS,shots}}$ shots per circuit for a total of $N_{\text{LCHS,samples}} \cdot N_{\text{LCHS,shots}}$ circuit executions.
We set $N_{\text{LCHS,samples}}=100$ and choose $N_{\text{LCHS,shots}}$ such that the total number of circuit executions are the same between Schr\"odingerization and LCHS.
For both methods, we apply the second-order Trotter formula with 20 Trotter steps.
In \fig{lchs_schro_comparison}, we plot the estimated value of $\expval{O}=\sum_{j=1}^{n} \frac{I-Z_j}{2}$ obtained by both methods.
We find that both methods obtain similar standard error when using the same total number of circuit executions.

\begin{figure}[ht]
    \centering
    \includegraphics[width=0.4\linewidth]{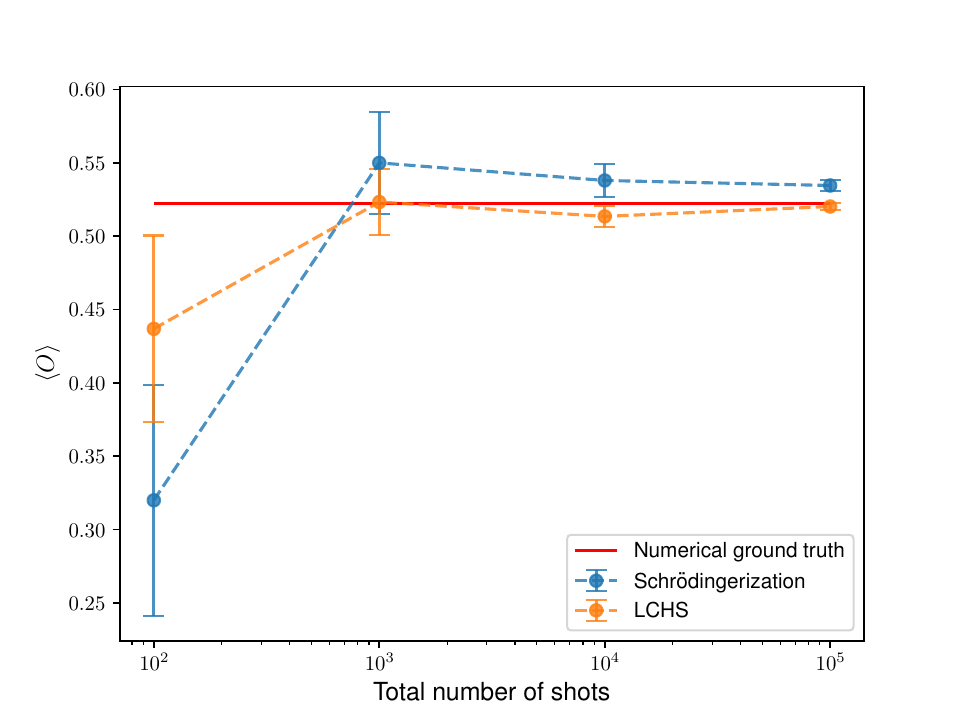}
    \caption{Estimate of $\expval{O}$ at $t=0.5$ with parameters $J=h=1$, $\gamma_z=0.1$. Error bars denote the standard error. The implementation of Schrodingerization uses $n_p=7$ ancilla qubits to represent $p\in[-R,R]$, where we set $R=25$. The implementation of LCHS uses the same choice of $R$ to perform Monte Carlo sampling, where $[-R,R]$ is discretized into 128 grid points.}
    \label{fig:lchs_schro_comparison}
\end{figure}